\documentclass[11pt]{article}
\input{setup}

\title{Status substitution and conspicuous consumption\footnote{Ghiglino: University of Essex (email: \emph{cghig@essex.ac.uk}), Langtry: University of Bristol (email: \emph{alastair.langtry@bristol.ac.uk}). We are grateful to Matt Elliott and Nicole Tabasso for helpful comments. This work was supported by the Economic and Social Research Council [award reference ES/P000738/1] and the Janeway Institute. Any remaining errors are the sole responsibility of the authors. }} 
\author{Christian Ghiglino \and Alastair Langtry}
\date{\today} 
\begin{document}
\maketitle
\begin{abstract}
\noindent This paper adapts ideas from social identity theory to set out a new framework for modelling conspicuous consumption. Agents derive utility from their consumption of a status good and from belonging to an identity group with high status good consumption. Importantly, these two sources of utility are substitutes. Agents also feel pressure to conform with their neighbours in a network. This framework can rationalise a set of seemingly conflicting stylised facts about conspicuous consumption that are currently explained by different families of models. In addition, our model delivers new testable predictions regarding the effect of network structure and income inequality on conspicuous consumption. 
 \\

\noindent \textbf{JEL:} \hspace{15mm} D63, D85, D91 \\ \noindent \textbf{Keywords:} \hspace{4mm} Social identity, social networks, interpersonal comparisons, norms,

\hspace{19mm}   discrimination,  centrality, income inequality, Keeping Up with the Joneses \\

\end{abstract}

\maketitle

\newpage
Economists have long acknowledged that some goods are consumed as much for the status they convey as for their actual usefulness, a phenomenon Veblen termed conspicuous consumption. For these \emph{status goods}, comparisons with others is understood to be a particularly important consideration in people's consumption decisions. To date, an impressive body of empirical work has documented a range of stylised facts regarding how people's consumption of status goods responds to the choices and incomes of others. 

First, an influential paper by \cite{charles2009conspicuous} documents that, holding their own permanent income fixed, people in the U.S. consume \emph{less} status goods (which they call `visible goods') when average income of those of the same race and who live in the same state is higher. \cite{khamis2012consumption} find similar behaviour for caste/religious groups in India (Fact 1). But second, people consume \emph{more} when average income of their geographic neighbours increases (\cite{kuhn2011effects}, Fact 2). And they consume more when the income of other social groups rises (\cite{bellet2018conspicuous}, Fact 3). Additionally, people's status good consumption is more sensitive to the incomes of others when they have higher income themselves (\cite{kaus2013conspicuous}, Fact 4).\footnote{A large part of the contribution in \cite{kaus2013conspicuous} is to extend the analysis in \cite{charles2009conspicuous} to a lower income country -- South Africa. He replicates the findings for Black South Africans (a large majority of the South African population), but not for White South Africans. But the finding regarding sensitivity is novel to \cite{kaus2013conspicuous}.}

Beyond these effects of changes in the level of income, \cite{charles2009conspicuous} documents that the impact of inequality can differ markedly across different groups (Fact 5). They find that higher dispersion of income within a person's reference group leads to lower status good consumption among Whites, but higher status good consumption among Blacks. Finally, there is widespread evidence that, conditional on income, some social groups spend more on status goods than others (\cite{charles2009conspicuous, kaus2013conspicuous, colson2022does}, Fact 6). 

These six facts are by no means an exhaustive description of how people consume status goods. Nevertheless, we contend that they capture many of the most important features of status good consumption that have been documented empirically. Our paper presents a simple model of social identity that explains these six stylised facts as a single phenomenon.  This stands in contrast to existing theory, which struggles to rationalize these behaviours as a coherent whole.

In our model, people belong to an identity group and derive utility from two sources: (1) directly from consuming status goods themselves,\footnote{Own consumption can bring both intrinsic utility and individual status.} and (2) from belonging to an identity with high average status good consumption, relative to that of other groups. Importantly, we assume there is some imperfect substitutability between these two sources. 

Belonging to an identity group also makes people feel pressure to conform to some reference point of status good consumption for their identity group (a `benchmark for normal'). We assume this reference point depends only on a person's neighbours in a social network, rather than their entire identity group. These two features -- deriving benefits from group status and the pressure to conform -- are both core components of social identity theory \citep{costa2015social}.

Our model sees agents gain utility directly from consuming the status good. This interpretation of status is the view of many sociologists \citep{coleman1994foundations} and of some economists \cite{Ocass2004exploring} and \citep{becker2005equilibrium}.\footnote{The origin of the desire for status, including evolutionary pressures, are explored in \cite{anderson2015desire}.} 
Within economics, it has proved especially popular in network economics (e.g. \cite{immorlica2017social, langtry2023keeping, ghiglino2010keeping}). But overall, a mainstream view in economics is that status comes only from income or wealth -- and that status goods are simply a costly signal \citep{bagwell1996veblen, glazer1996signaling, charles2009conspicuous}. 

Note that most of the literature uses status consumption and conspicuous consumption interchangeably (see discussion in \cite{Ocass2004exploring}). We adopt the view that conspicuous consumption, that is, consumption of a good in excess of its usefulness, is the result of the good providing status and prestige.

Agents also gain status from belonging to a group with high-status members. Here, this is one with high status good consumption. But group status is relative -- it depends on the group's average consumption \emph{relative to} that of the other group(s). This idea that status is relative can be traced back at least to Max Weber's (\citeyear{weber1918}) original work on status groups. It remains a firm feature of thinking about status -- regardless of whether one adopts the direct status view or the signalling view.\footnote{In \Cref{OA:abs status}, we consider the alternative case where status is absolute. That is, where it depends only on the consumption of an agent's own group.}

On the way to explaining the six stylised facts set out above, we provide a simple characterisation of behaviour. Equilibrium consumption of status goods depends critically on a generalised version of Bonacich centrality \citep{katz1953new, bonacich1987power} -- a well-known measure of network centrality. But here, both an agent's own centrality and the average centrality of each identity group are important. An agent's status good consumption increases linearly in her own centrality (a feature common to many networks models). But it also decreases in the average centrality of her own identity group and increases in the average centrality of other identity groups. In line with existing work, centrality is increasing in the strength of network connections, in an agent's own income and in the income of those who share her identity. 

The contrasting impacts of own centrality and group-average centrality is driven by the two sources of status in our model -- own consumption and group-average consumption -- and the fact that they are imperfect substitutes. Higher own centrality raises an agent's status good consumption directly. But higher centrality for others in her group raises \emph{their} consumption -- which raises group-average consumption, in turn raising the group's status. This higher group status reduces the marginal benefit of an agent's own consumption, inducing her to consume less. Because group status is inherently relative, higher centrality in other groups has the opposite effect. Higher centrality, and hence higher consumption, in other groups lowers the agent's group status. In turn, this pushes her to consume more status goods herself.

It is precisely these contrasting impacts that explain the first four stylised facts. When incomes for some agents in a given identity group rise, they consume more status goods -- the traditional direct response to an increase in own income. This then creates two spillover effects. First, it increases centralities for others in that identity group. 
This effect pushes status good consumption of those others in the identity group \emph{up}. 
Second, it increases average consumption of that identity group -- which raises group status for that identity group. In turn, this pushes others within the identity group to consume \emph{less} status goods -- because it reduces the marginal benefit of an agent's own status good consumption. 

Which of the two spillover effects wins out depends on how closely connected an agent $i$ is to those whose incomes increased. When those whose incomes increased are close neighbours (as in \cite{kuhn2011effects} and \cite{khamis2012consumption}) we should expect the first effect to win out -- so $i$ consumes more status goods (Fact 2). In contrast, when $i$  is not closely connected to those whose incomes increased (as in \cite{charles2009conspicuous}) we should expect the second effect to win out -- delivering the seemingly counter-intuitive result that she consumes fewer status goods (Fact 1). 

Notice that for other identity groups, only the second effect operates. There is no effect on own centrality in different identity groups. But higher average consumption in one identity group reduces the group status in all other groups. In turn, this pushes others in different identity groups to consume \emph{more} status goods (Fact 3). The explanation for Fact 4 is even more direct. Centrality acts multiplicatively on all other terms in our characterisation of equilibrium status good consumption. And centrality is increasing in own income. So higher own incomes makes status good consumption more sensitive to changes in other primitives -- including others' incomes.

We then study the effects of changing inequality. Here, we consider a setting with a number of separate communities and homogenous incomes within each community. The impact of redistributing income across different communities depends critically on the relative densities of network connections within the affected communities. This highlights the role that network connections play in mediating the impact of inequality. Here, we present some stylised evidence of a non-monotonic relationship between income and network densities. This relationship, coupled with well-documented differences in average incomes across different groups, delivers sharply differing effects of inequality on different identity groups (Fact 5). 

Finally, we extend the model to incorporate some exogenous component to group status -- which we call prestige. This captures the idea that many other factors beyond status good consumption may influence an identity group's status. One important determinant of a group's prestige may be discrimination it faces within a society \citep{oldmeadow2010social, fox2015denying}. In our model, groups with lower prestige consume more status goods (all else equal). This is because agents in these groups use status good consumption to compensate for lower prestige. Differences in discrimination/prestige directly induce differences in status good consumption across groups (Fact 6). And higher consumption of status goods would lead to lower consumption of other (unmodelled) goods -- including capital goods. So our model draws a novel link between discrimination and lower consumption of capital goods.

The rest of the paper is organised in the usual way. \Cref{sec:literature} briefly reviews related literature. \Cref{sec:model} presents the model. \Cref{sec:results} presents the main results. \Cref{sec:stylised facts} shows how our model explains the stylised facts. Then \Cref{sec:inequality} examines the impact of inequality. \Cref{sec:discrimination} extends the model to include an exogenous component of group status, and discusses the effects of discrimination. \Cref{sec:conclusion} concludes. All proofs are in \Cref{appendix:proofs}.

\section{Literature}\label{sec:literature}
Our paper sits within a large literature on interpersonal comparisons and conspicuous consumption.\footnote{It spans disciplines, including marketing, sociology, social psychology (see \cite{dubois2017social} and \cite{dubois2021psychology} for recent reviews), and economics (where it dates at least to \cite{veblen1899}, with other early contributions from  \cite{duesenberry1949income,festinger1954theory, frank1985demand, frank1985choosing} and \cite{van1985relativity}). \cite{bursztyn2017social} provide an interesting survey of the wider topic of social pressure.} 
Within economics, there are two broad views on why people engage in conspicuous consumption. The first assumes that people care about their income or wealth relative to others -- and conspicuous consumption is used as a \emph{signal} of high income/wealth. \cite{charles2009conspicuous} and \cite{khamis2012consumption} use this signalling approach to explain the first stylised fact (Fact 1) -- when others get richer, a given agent moves down the income ranking, and so has fewer people she needs to signal she is richer than (their model follows \cite{glazer1996signaling} very closely). However, with no explicit notion of `neighbours', these models struggle to explain the second stylised fact -- that people consume more when their neighbours' incomes rise. Additionally, because agents are consuming to differentiate themselves from people poorer than themselves, everyone's consumption decisions can be sensitive to incomes in the lowest part of the distribution. 
Our model can speak to Fact 1 without this sensitivity.

The second view assumes that people make social comparisons directly over consumption -- so people gain status from conspicuous consumption itself. The core idea is that people like high consumption relative to others, not just in absolute terms -- often called a ``Keeping up with the Joneses'' motive. Earlier work assumed people either care about their consumption relative to the population average, or to their rank within the population. Important contributions to these two approaches include \cite{abel1990asset, clark1996satisfaction, ljungqvist2000tax}, and \citep{frank1985choosing, frank1985demand, frank2001luxury, hopkins2004running} respectively.\footnote{More recent contributions include \cite{luttmer2005neighbors, heffetz2011test, frank2014expenditure, drechsel2014consumption, alvarez2016envy, jinkins2016conspicuous, bertrand2016trickle},\cite{de2020consumption} and \cite{hopkins2023cardinal}} 
Our paper is most closely related to a recent strand within this literature, started by \cite{ghiglino2010keeping}, that uses explicit networks to allow for very rich patterns of social comparisons (more recent contributions to this strand include \cite{immorlica2017social,ushchev2020social,langtry2023keeping} and \cite{bramoulle2022loss}). 
These `Keeping up with the Joneses' models successfully explain the second stylised fact. But they cannot explain the first. This is because they embed a desire to `Keep up with the Joneses', which makes agents' consumption choices strategic complements -- an increase in some people's income must (weakly) increase \emph{everyone's} consumption.

Our paper explains a set of so-far difficult to reconcile stylised facts in a single framework. It does so by building a new type of model -- one motivated by social identity, rather than signalling or `Keeping up with the Joneses' motives. 
In doing so, we also relate to a strand of literature that examines the link between inequality and conspicuous consumption. By and large, formal models provide ambiguous results on the nature of the link. In \cite{charles2009conspicuous}, a regressive redistribution of income (from a poorer agent to a richer one) increases total conspicuous consumption only if consumption is convex in income -- but this is not a consistent property of observed consumption data \citep{heffetz2011test}.\footnote{Other versions of the signalling model make similar predictions \citep{ireland2001optimal, heffetz2011test, moav2012saving, jinkins2016conspicuous}.}

The present model assumes cardinal utilities. \cite{bilancini2012redistribution} show that the relationship between conspicuous consumption and inequality can be sensitive to the choice of whether to use ordinal or cardinal utility. \cite{hopkins2023cardinal}, who extends \cite{hopkins2004running} to cardinal games, reaches a  similar conclusion and shows that cardinal games are more general and offer a better description of the evidence.  

As in this literature, we do not present an ironclad relationship between inequality and conspicuous consumption. But we do provide novel predictions regarding how the nature of the relationship depends on properties of the network. Notably, it allows us to explain our fifth stylised fact -- that higher dispersion in income within a person’s reference group leads to lower conspicuous consumption among Whites, but higher conspicuous consumption among Blacks. It does so without an appeal to differences in underlying preferences across groups, but instead relies on well-documented differences in average income.

Finally, our paper relates to the burgeoning literature on social identity theory within economics (see \cite{costa2015social} for a survey). The core idea in these models is that agents are part of an `identity' group, and gain some status from belonging to the group, but also bear some costs from deviating from some benchmark for that group (either real differences, or in the agents' own perception of matters). This setup was formalised in seminal contributions from \cite{akerlof2000economics} and \cite{shayo2009model}.\footnote{A number of recent papers extend or apply Shayo's approach; see for example, \cite{klor2010social, gennaioli2019identity, grossman2021identity, lindqvist2013identity, atkin2021we}, and \cite{shayo2020social}. \cite{ghiglino2024endogenous} consider how network connections affect the choice of identity (in contrast, our paper treats identity as fixed and considers how the network affects consumption choices). 
}

Our main contribution here is to show how models of social identity can help resolve a puzzle regarding conspicuous consumption behaviour -- something these models have not been used for to date. However, we abstract away from the choice of identity group -- which is often a key focus in this literature. This is both in the interest of tractability and because it is not needed to explain the stylised facts we focus on. From a more technical standpoint, aside from this, our utility formulation closely follows \cite{shayo2009model}. However, departing from most of the literature, but closer to \cite{atkin2021we}, we make the benchmark (often called   `prototype') as well as status both endogenous.

\section{Model}\label{sec:model}
\paragraph{Agents and endowments.} There are $J$ agents, with typical agent $j$. Each agent is endowed with real income $w_j \in [1,\infty)$ and a fixed \emph{identity}, from two possible identities: $\theta_j \in \{A,B\}$. Agents are embedded in a weighted and directed network represented by a $J \times J$ non-negative matrix $G$. We say that agent $k$ is a \emph{neighbour} of $j$ if $G_{jk} > 0$, and assume that $j$ is not her own neighbour (so $G_{jj} = 0$ for all $j$). Let $I(\theta) = \{ k : \theta_k = \theta \}$ be the set of agents who have identity $\theta$. 

\paragraph{Consumption game.} Agents simultaneously choose \emph{consumption} of a status good, $x_j \geq 0$. 

\paragraph{Reference point.} Each agent $j$ has a \emph{reference point}, $R_j = \alpha_j \sum_{k \in I(\theta_j)} g_{jk} x_k$, where $g_{jk} = \frac{G_{jk}}{\alpha_j}$ and $\alpha_j = \sum_k G_{jk}$; a weighted average of the status good consumption of her neighbours who share her identity. 

\paragraph{Preferences.} The agent $j$ receives utility benefits from her own consumption, $x_j$, and from the status of the identity group she identifies with\footnote{The assumption that status goods have both intrinsic and status benefits is also present in \cite{immorlica2017social}.}. The status term is given by the average consumption of her identity group \emph{relative to} the other identity group. For simplicity, we model this relative term as the ratio of average consumption of each group: $Y_\theta = \frac{\mathbf{x}_\theta}{\mathbf{x}_{-\theta}}$, where $\mathbf{x}_{\theta} = \frac{1}{|I(\theta)|} \sum_{k \in I(\theta)} x_k$ for $\theta \in \{A,B\}$. Importantly, we assume that own consumption and identity status are imperfect substitutes: the marginal benefit of each is decreasing in the level of the other.

Agent $j$ also experiences \emph{dissonance} costs when her consumption differs from her reference point, and these costs are convex. This is standard in identity theory. Note that $j$'s status benefits depend on everyone who shares her identity, but her reference point depends only on her neighbours who share her identity. 

Finally, there is also a convex cost of own consumption, with both absolute and marginal cost decreasing in own income. This cost function captures spending on other goods (outside of our model) in a reduced form way; implicitly assuming concave benefits from consuming the unmodelled goods (see discussion in \Cref{sec:model_discussion}).
We assume agents have the following preferences:
\begin{align}\label{eq:prefs}
    u_j = \underbrace{ S(x_j, Y_{\theta_j}) }_{\substack{\text{Intrinsic utility} \\ \text{and Status}}} - \underbrace{ \vphantom{\frac{\beta_j}{\beta_j}} \frac{1}{2} \beta (x_j - R_j)^2 }_{\text{Dissonance}} - \underbrace{ \frac{1}{2 w_j} x_j^2 }_{\text{Cost}}, 
\end{align}
with $\beta>0$ and $S'_x > 0, S'_Y > 0, S''_{xY} < 0$, for $x \in [0,X]$ and $Y \in \left[\frac{1}{Z},Z \right]$ for some $X > 0, Z \in (0, \infty)$. 
We provide a sufficient characterisation of $X$ and $Z$ in \Cref{sec:machinery}. In the interest of tractability, we  assume that the utility from own consumption is linear in $x_j$, so $S'_x$ is not a function of $x_j$.
A simple example of a function that meets these assumptions is $S(x,Y) = \eta \, x + Y - \gamma \, x \, Y$, where $\eta > 0$ and $\gamma < \min\{\frac{\alpha}{Z}, \frac{1}{X}\}$.

For tractability we also assume that agent $j$ ignores her own impact on $Y_j$ when choosing her consumption. When the number of agents is large, $j$'s impact on $Y_j$ is small, so this assumption has less and less `bite' as $J$ grows. An alternative approach would be to assume an infinite number of agents, so $j$'s impact on $Y_j$ is exactly zero. Such an approach would require extra notation without adding insight.

\paragraph{Information, Timing, Solution Concept.}
All of the primitives -- the network, incomes, functional forms, the value of beta -- are common knowledge. All agents move simultaneously. So we look for Nash Equilibrium of the game. 

\subsection{Machinery \& Assumptions}\label{sec:machinery}
The preferences in Eq (1) put our model into the class of linear-quadratic games on a network (see \cite{jackson2015games,bramoulle2016games} for reviews). In particular, the quadratic dissonance term is similar in spirit to \cite{ushchev2020social}. As is standard in this class of games, some variant on Bonacich centrality will play an important role in equilibrium behaviour.

Our variant on Bonacich centrality is intuitively much the same as the standard definition, but uses only within-identity links in the network, and these links are weighted by the importance of the dissonance function, $\beta$, and by agents' incomes, $w_j$. The reason that only the within-identity links play a role is because we assumed that in both the status and dissonance functions, an agent only cares about those who share her identity. Formally, define a network $H$ where $H_{jk} = \frac{\beta w_j}{\beta w_j + 1} \widehat{G}_{jk}$ for all $j,k$, with $\widehat{G}_{jk} = G_{jk} \cdot \mathbf{1}\{\theta_j = \theta_k\}$ for all $j,k$ (where $\mathbf{1}\{\cdot\}$ is the usual indicator function). It will be important to bear in mind that the network $H$ is a function of $\widehat{G}$, agents' incomes, and of $\beta$.

\begin{defn}[Modified Bonacich Centrality] \label{defn:bonacich centrality}
The centrality of agent $j$ is:
\begin{align}\label{eq:bonacich centrality}
C_j = \sum_{k} \left[(I - H)^{-1}\right]_{jk} \cdot \frac{w_k}{\beta w_k + 1}
\end{align}
\end{defn}

In the interest of brevity, we will just refer to this as `centrality'. It will also be convenient to define the average centrality of agents with a given identity. Formally, let $\mathcal{C}_\theta = \frac{1}{|I(\theta)|} \sum_{k \in I(\theta)} C_k$ be the average centrality of agents with identity $\theta$. 
And to ensure centrality is well-defined (and that an equilibrium exists) impose the eigenvalue standard to linear-quadratic games \citep{ballester2006s, bramoulle2016oxford}. 
We also need the marginal benefit of consumption $S'_x(\cdot)$ to be bounded.

\begin{ass}\label{ass:eigenvalue}
    \normalfont{(i)} $|\lambda_1| < 1$, where $\lambda_1$ is the largest eigenvalue modulus of the matrix $G$, 
    
    \noindent \normalfont{(ii)} $\rho < \infty$, where $\psi \equiv \max_j \{C_j\} < \infty$ and $\rho \equiv S'_x(\frac{1}{\psi (1+\beta)})$.
\end{ass}

Part (ii) ensures that consumption is finite for all agents. It also ensures that the ratio of average consumptions, $Y_\theta$ is also finite. Together, this allow us to place our assumptions on the function $S(x,Y)$ only for a finite range. Specifically, for $x \in [0,X]$, where $X\geq \psi \rho$ and $Y \in [\frac{1}{Z}, Z]$ where $Z \geq \psi (1+\beta)$.\footnote{This is important, because the assumptions we impose on $S(x,Y)$ cannot hold for an infinite range. Note that our characterisation of the range in which our assumptions $S'_x > 0, S'_Y > 0, S''_{xY} < 0$ must hold could be tightened if we imposed some additional, but still mild, assumptions. We discuss this in \Cref{OA:additional_proofs}}
We impose Assumption \ref{ass:eigenvalue} throughout the paper and will not restate it for each result.

\subsection{Discussion}\label{sec:model_discussion}
Before going further, it is helpful to briefly discuss three key features of the modelling. First, the theory of where status comes from. Second, that we have only one good. And third, the dissonance term: its motivation and relationship to existing work on social identity.

\paragraph{Status substitution.} Our model is one where consuming the status good brings utility directly, rather than being a signal for unobservable wealth. We discussed this in the introduction, and will not repeat ourselves here. Nevertheless, it is important to note that we need status good consumption to be visible. This is the same as in signalling models, but the reason is different. Here, agents must see their neighbours' consumption in order to form their reference point. 

Importantly, we assume (imperfect) substitutability between the direct benefits from own consumption and the status benefits derived from social identification. Such a trade-off between the self and the group identity is at the heart of social identity theory \citep{TajfelTurner1979}. There is extensive empirical evidence for substitution, in particular that individuals with low personal self-esteem or competence increase reliance on group identification. Seminal papers include \citep{AbramsHogg1988,Hogg2000,CrockerLuhtanen1990}) while recent empirical studies confirming this include \citep{CAMBON201824, GonzalesBacken2015}.


\paragraph{Single good.} To help present clean results, our model has only one good, and uses a quadratic cost term to capture the opportunity cost of foregoing other (unmodelled) goods. This cost term depends explicitly on income, as the level of income affects the marginal opportunity cost of other goods. This income term implicitly accounts for (exogenous) prices -- as these also affect the marginal opportunity cost. 

We can also view this setup as a second order approximation of a two-good model with exogenous prices and budget. Formally, there is a status good $x$, a composite (non-status) good $z$, and additively separable utility $u(x,z) = v(x) + s(z)$, where $v(\cdot)$ and $s(\cdot)$ are increasing, strictly increasing and concave. Let $p$ denote the relative price of the status good, $W$ denote income. Substituting in the (binding) budget constraint $px + z = W$ and then taking a second order Taylor approximation of $u(x,z)$ delivers utility as a quadratic function of $x$. Finally, notice that \Cref{eq:prefs} -- while written in a particular way to capture notions of status and dissonance -- is ultimately a quadratic function of $x$. 

Note that our model -- even viewed as an approximation to a two-good case -- treats prices as fixed. It takes a partial equilibrium view that ignores the supply side of the economy. We do this because general equilibrium effects in social comparisons model are known to be extremely complex \citep{ghiglino2010keeping}, and more importantly because they would not bring additional insights regarding the type of questions we address in this paper. 

\paragraph{Dissonance.} Our utility formulation in \Cref{eq:prefs} closely follows \cite{shayo2009model} -- in that its two core components are the status benefits and dissonance costs. Agents bear dissonance costs from departing from the `benchmark' for their identity group. Here, this `benchmark' is called a reference point.\footnote{Note that in the social identity literature, the `benchmark' is typically called a \emph{prototype} (including in \cite{shayo2009model, shayo2020social}), rather than a `reference point'. The `reference point' terminology is more common in the literature on social comparisons, and stems from that literature's more explicit link to behavioural economics.} 
\cite{shayo2009model} focuses on how people \emph{choose} their identity. So for tractability he treats the dissonance costs as depending only on the identity choice -- not on consumption. 

In contrast, we assume dissonance costs depend on consumption, and that the reference point (the `benchmark') is also determined by endogenous consumption choices.
Further, we allow the reference point to depend on an agent's exact position in a network. This follows the now-common notion that people compare themselves to those they actually interact with (e.g. \cite{ghiglino2010keeping, immorlica2017social, bramoulle2022loss}). The trade-off is that we treat identity as exogenously fixed in order to keep the model tractable and ensure clean results.

A limitation of our reference point is that it only depends on neighbours with the same identity. Intuitively, this is a highly stylised way of capturing the view that an agent's benchmark for normal -- that they want to conform to -- is set predominantly by those in their own identity group. More prosaically, it is important for tractability.

Nevertheless, we follow the standard view in social identity models that people want to conform: they dislike being either above or below the reference point. This is motivated by significant evidence from social psychology, especially regarding descriptive norms \citep{cialdini1990focus, cialdini2004social}.
But it stands in contrast to the standard assumption in the social comparisons literature that people want to get ahead: they \emph{dislike} being below the reference point. Ultimately, changing a social comparisons style desire to get ahead would not alter the key insights of our model. We show this in \Cref{OA:KUWJ}. 


Before showing how our model rationalises the stylised facts set out in the introduction, we first analyse the model on its own terms. We characterise equilibrium behaviour, and show how it depends critically on an agent's own centrality, and the average centralities of the identity groups. Then we examine how the primitives of the model impact centralities.

\section{Results}\label{sec:results}
We start by solving the model: that is, providing an exact characterisation of how agents behave in equilibrium.

\begin{prop}\label{prop:1}
There exists a unique equilibrium. For all $j$, equilibrium consumption is:
    \begin{align}\label{eq:closed form solution}
    x_j^* = S'_x(Y_{\theta_j}^*) \cdot C_j, \text{  where  } Y_{\theta_j}^* \text{ solves } Y_{\theta_j} = \frac{S'_x(Y_{\theta_j}) \mathcal{C}_{\theta_j}}{S'_x(\frac{1}{Y_{\theta_j}}) \mathcal{C}_{-\theta_j}}
\end{align}
\end{prop}

Centrality is the key determinant of equilibrium consumption. As in other linear-quadratic games, higher own centrality increases an agent's consumption. Here, this is because agents with high centrality are more connected to others who themselves have high consumption -- giving them a high reference point. In turn this pushes them to consume more in an effort to conform to the `benchmark for normal' set by the reference point. 

There is also high status from belonging to a group where the other members have high consumption (which is driven by them having high centrality). But status from the group is imperfectly substitutable with the status from own consumption. So when an agent belongs to a high status group, she has a lower marginal benefit of own consumption -- and so she consumes less. The substitutability between the two sources of status is key. When an agent belongs to a high status group, she feels less need to create status for herself through her own consumption.

But group status is relative. High consumption in the \emph{other} identity group -- driven by high centrality in that group --  reduces the group status. By identical logic, this raises consumption. The relative notion of status drives cross-group effects: what happens in one identity group affects behaviour in the other -- even though we assumed away any direct benchmarking against one another through the reference point.
The following result formalises this discussion.

\begin{rem}\label{rem:1}
$x_j^*$ is strictly increasing in $C_j$ and in $\mathcal{C}_{- \theta_j}$, and strictly decreasing in $\mathcal{C}_{\theta_j}$, for all $j$.
\end{rem}

This links equilibrium behaviour to centrality. However, note that the comparative static on $C_j$ is just a thought experiment. Changes to primitives of the model will change the centrality of multiple agents, so changing $C_j$ while holding all else constant is not possible in practice. The next natural question is to ask what drives centrality. Unsurprisingly, high centrality in our model comes from a combination of strong connections to neighbours (especially to those who are themselves well connected), connections to people with high income, and having high income oneself.

\begin{rem}\label{rem:2}
$C_j$ is increasing in $w_\ell$, and in $G_{\ell m}$, for all $j,\ell, m$. These results are strict if and only if $j$ and $\ell$ share an identity and there is a walk from $\ell$ to $j$.
\end{rem}

More and stronger links raise an agent's reference point -- the benchmark she tries to conform to. So she raises her consumption in response. In turn, this then pushes up her neighbours' reference points, which raises their consumption, and so on. The effect ripples out through the network. Similarly, higher income makes consumption cheaper for an agent (importantly, at the margin) -- so she chooses to consume more directly. This then raises her neighbours' reference points, and so ripples out through the network in much the same way as with stronger links. In both cases, the ripple effect is contained within the affected identity group -- other identity groups are unaffected.

The 'containment' of the ripple effect is driven by the fact that agents only consider those who share their identity when setting their reference point. It captures the idea that I only feel pressure to conform to `people like me'. Those with a different identity are seen as different, and so not conforming with them does not present any issues.

So far, \Cref{rem:1} links centralities to consumption decisions, and \Cref{rem:2} links incomes and network structure to centralities. But the overall link between incomes and consumption can be ambiguous. When incomes within an agent's own identity group change, there are two effects. First, an increase in someone else's income increases her own centrality. And second, it increases the average centrality of her group. These push in different directions (\Cref{rem:1}). In contrast, an increase in income among agents in the \emph{other} identity group has a clear effect. It must increase $j$'s consumption. This is because it reduces her group status without any impact on her own centrality. The following result states this formally. 

\begin{prop}\label{prop:slutsky}
(i) if $\theta_j \neq \theta_k$, then $x_j^*$ is strictly increasing in $w_k$.
(ii) if $\theta_j = \theta_k$, then
\begin{align}\label{eq:slutsky}
    \frac{d x_j^*}{d w_k} = 
    \underbrace{ S'_x(Y^*_{\theta_j}) \cdot \frac{d C_j}{d w_k} \vphantom{\Bigg|_T}}_{\text{own centrality effect} } +
    \underbrace{ C_j \cdot \frac{d S'_x(Y^*_{\theta_j})}{d Y^*_{\theta_j}} \cdot \frac{d Y^*_{\theta_j}}{d \mathcal{C}_\theta}  \cdot \frac{d \mathcal{C}_\theta}{d w_k} }_{\text{ average centrality effect} }
\end{align}
\end{prop}

This shows that spillovers across identity groups are clear -- higher income in one identity group always drives higher consumption in the other. And for changes within an identity group, it decomposes the two contrasting forces. This helps us to see when the own centrality effect will dominate -- pushing consumption up -- and when the average centrality effect will dominate -- pushing consumption down. But it is difficult to say in general exactly what network structure will lead one effect versus the other to dominate. 

Well-known intuitions regarding Bonacich centrality style measures suggest that a change to incomes (or other primitives) has a larger impact on agents closer in the network to the change. So the key idea is that for agents located `close' to the change, the effect on own centrality will dominate the effect on average centrality. And that the reverse applies for agents located `far from' the change.

When we turn to explaining the stylised facts in the next sections, we will focus on a stark benchmark setting where agents are split up into disconnected `communities'. This will allow us to separate each channel cleanly. And it gives us a precise, if stylised, way of thinking about `near' and `far' in the network. A pair of agents who share a community are `near' to one another in the network, and a pair in different communities are `far apart'.

\section{The effects of changes in income}\label{sec:stylised facts}
In the introduction, we set out a series of stylised facts about how changes in incomes feed through into how people choose to consume status goods. And we noted that existing models struggle to rationalise them all as a coherent whole. The aim of this section is to show how our model fits the first four of these six facts simultaneously. It will use the results from \Cref{sec:results}. The next two sections will show how our model fits the fifth and sixth stylised facts respectively, as well as building some additional results related to them.



As noted at the end of \Cref{sec:results}, we consider a set of disconnected `communities'. This allows us to isolate the channels identified in \Cref{prop:slutsky}.

\begin{defn}[Communities network]\label{defn:communities_network}
The network $\widehat{G}$ consists of $N$ strongly connected components,\footnote{A \emph{walk} in $\widehat{G}$ from $j$ to $k$ is a sequence of agents $j, j',...,k',k$ such that $\widehat{G}_{j'k'}>0$ for all $j',k'$ in the sequence. A \emph{strongly connected component} is a set of agents such that for all agents $j,k$ in the set, there exists a walk from $j$ to $k$.} 
with typical component $n$, and where $N$ is sufficiently large. 
All components contain the same number of agents, and there are an equal number of components of each identity.  Call each component a \emph{community}.
\end{defn}

Changes in community $n$ will only affect centralities within that community, and will not spill over into other communities. This means a change in one community affects other communities only through the impact on average centrality. Further, assuming that there are many communities guarantees that a change in community $n$ has a larger impact on the centrality of the agents in that community than it does on the average centrality of all agents. In summary, following a change in community $n$, the `own centrality' channel dominates within community $n$, and the `average centrality' channel dominates in all other communities.

Before stating our result, it is helpful to quickly recap the stylised facts that relate to changes in income. This will help make the link between them and the behaviour predicted by our model easier to see.

\begin{enumerate}
    \item In the US [resp. India], people consume \emph{less} status goods when the average income of others of the same race in their state [others of the same caste/religious group in their district] is higher \citep{charles2009conspicuous, khamis2012consumption}.
    \item In the Netherlands, people consume \emph{more} when the average income of their geographic neighbours increases \citep{kuhn2011effects}.
    \item In India, low caste households' consumption of status goods is \emph{higher} when Upper Caste households are \emph{richer} \citep{bellet2018conspicuous}.
    \item In South Africa, people's consumption of status goods is \emph{more sensitive} to changes in income of their reference group when they are richer \citep{kaus2013conspicuous}.
\end{enumerate}

With these facts in mind, we can now state our result.

\begin{prop}\label{prop:2}
    Consider a communities network. Suppose income increases for some agents with identity $\theta$ in community $n$. \\
    (i) \ \ Consumption decreases for all agents with identity $\theta$ in communities \emph{other than} $n$, \\
    (ii) \ Consumption increases for all agents with identity $\theta$ in community $n$. \\
    (iii) Consumption increases for all agents with identity $\theta' \neq \theta$. \\ 
    (iv) \ For all agents $j$ in communities \emph{other than} $n$, the marginal impact on $j$'s consumption is increasing in $j$'s own income.
\end{prop}

The mapping from \Cref{prop:2} into the stylised facts is fairly direct. An increase in incomes in other communities corresponds to higher average income of others of the same race in the same state, as in \cite{charles2009conspicuous}.\footnote{We use the term `income' in our model for convenience. Where appropriate, this should be taken to mean the `permanent income' examined by \cite{charles2009conspicuous} in their empirical analysis.} 
An increase in incomes within an agent's own community corresponds to raising her neighbours' incomes, as in \cite{kuhn2011effects}. Note that \Cref{prop:2}(ii) applies only to agents who share an identity, but \cite{kuhn2011effects} does not consider identity in their analysis. Nevertheless, strong homophily and sorting along geographic lines are well-documented phenomena within economics, sociology \citep{mcpherson2001birds,schelling2006micromotives} and other disciplines. So it is very likely that at least some of an agent's geographic neighbours do in fact share her identity. 

The third stylised fact is even more direct: an `increasing marginal impact' is simply a formalisation of the `greater sensitivity' found by \cite{kaus2013conspicuous}. Finally, an increase in incomes of one identity group (Upper Caste households) leads to an increase in visible/status goods consumption of other identity groups (low caste households). This is the fourth stylised fact.

From a more technical standpoint, \Cref{prop:2} follows directly from \Cref{prop:1} when focusing attention on the communities network. 
%
Importantly, \Cref{prop:1} illuminates the mechanism that drives this behaviour. An increase in incomes for some agents with identity $\theta$ in community $n$ raises their consumption directly. This then raises the reference points of their neighbours, raising the benchmark these neighbours are trying to conform to -- and hence raising their consumption (Fact 2). This effect ripples out through the whole community (see discussion following \Cref{rem:2}). This then raises the average consumption of agents with identity $\theta$. So agents with identity $\theta$ gain more status from belonging to that identity group. So those with identity $\theta$ in other communities feel less need to seek status through their own consumption of status goods (Fact 1). Conversely, agents in other identity groups now have lower status -- because status is a relative concept -- and so feel a greater need to consume status goods (Fact 4).

But for agents in community $n$, the changes are concentrated `close' to them in the network. So the pressure to conform to a higher benchmark is strong relative to the gains from their whole identity group having higher status. It is this pressure from a higher benchmark that pushes up consumption -- and this is the force that dominates inside community $n$ (Fact 1).
Finally, while Propositions \ref{prop:1} and \ref{prop:2} are clear on how the third stylised fact is captured mechanically, the headline intuition is related to the way income affects the costs of consumption. In particular, agents with higher income have a lower marginal cost of consumption, and this marginal cost rises more slowly. When the marginal status benefits of consumption rise, as happens when $Y_{\theta_j}$ decreases, it should not come as a surprise that agents increase their consumption in response. But by how much they adjust their consumption depends on how quickly the marginal cost rises as consumption rises. It is precisely higher income agents whose marginal costs rise more slowly, and so who make greater adjustments.

\section{The effects of changes in inequality}\label{sec:inequality}
We now turn to our fifth stylised fact -- a puzzling empirical regularity documented by Charles et al. (2009) regarding the effects of income inequality.
Specifically, they find that higher income inequality among Whites in a given U.S. state correlates with lower consumption of visible (status) goods by Whites, even when controlling for individual income. But for Blacks, the same rise in inequality leads to \emph{higher} consumption of visible (status) goods, again holding individual income constant.

In this section, we present an explanation for this stylised fact. There are three steps. First, we present a formal result linking the qualitative impact of changes in inequality to the density of social links in communities. 
Second, we argue that there is convincing evidence that density is linked to income and has a U-shaped relationship with income. Finally, we show that coupling this U-shaped relationship with our formal result can explain the stylised fact discovered by Charles et al. (2009). 

\paragraph{Inequality and the role of density.} To help clarify things, we assume that all agents in a given community $n$ share the same income, expressed formally as \( w_i = w_n \) for all \( i \in n \). With this assumption in place, we define the density of a community, which will be a crucial metric for studying the effects of changes in between-community inequality.

\begin{defn}[Community Density]\label{defn:island density}
The density of community \( n \) is:
\begin{align*}
    D_n = \sum_{j,k \in n} \left[ (I - \widehat{G})^{-1} \right]_{jk}
\end{align*}
\end{defn}

Readers familiar with Bonacich centrality will recognize that our definition of density sums the centralities of all agents in the community -- using the standard definition of Bonacich centrality that focuses solely on network connections, rather than our variant that incorporates income. While this definition diverges from the conventional notion of network density, it captures an important aspect of the overall connectivity within a community.\footnote{Conventionally, density is simply the sum of all links (normalised by the number of agents). Here, this would be: $\frac{1}{|n|} \sum_{j,k \in n} \widehat{G}_{jk}$. Both this standard measure and our measure are strictly increasing in the weight of each link.}

We focus on a natural form of redistribution -- a transfer of income from one community to another while maintaining constant average income across society. Such a transfer decreases income inequality if the receiving community is poorer than the donor community and increases inequality if the recipient is wealthier. 

\begin{prop}[Inequality]\label{prop:inequality 1}
Consider two communities, \( n \) and \( n' \), both sharing the same identity \( \theta \). When income is transferred from \( n \) to \( n' \) without changing their income rankings, we observe:
\begin{itemize}
    \item[(i)] Consumption decreases in community \( n \) and increases in community \( n' \).
    \item[(ii)] Consumption in all \emph{other} communities with identity $\theta$ increases if and only if \( D_{n'} < D_n \).
\end{itemize}
\end{prop}

The first part of this proposition is intuitive: when income rises, consumption increases, and when it falls, consumption decreases. However, the effect of this transfer on the sum of status good consumption in communities $n$ and $n'$  depends on their relative densities. If the receiving community ($n'$) has a denser network, the increase in consumption there outweighs the fall in the donor community ($n$). This is because a denser network makes the receiving community's consumption more sensitive to a change in income. This enhances the status of the identity group. 

The other communities with identity $\theta$ -- those not party to the transfer -- respond to this enhanced group status by reducing their consumption.\footnote{Importantly, this reduction cannot fully offset the initial rise in consumption among the two affected communities. So we also have that \emph{total} consumption (across all communities) rises if and only if \( D_{n'} > D_n \). \label{ft:inequality_total_effect}} 
In other words, with higher status from their group, agents feel less need to consume status goods themselves. This is due to the imperfect substitutability between own status good consumption and group status.

Notice that this depends only on the \emph{relative} incomes of the two communities who are party to the transfer -- not on where in the overall income distribution these two communities sit. So it applies equally to changes concentrated at the top of the income distribution (e.g. the very rich pulling away from the merely quite rich) as it does to changes between the rich and the poor.

\paragraph{The link between density and income.} The connection between this proposition and the observed stylised facts hinges on an assumption that density (a la Definition \ref{defn:island density}) follows a U-shaped pattern relative to income. That is, both low- and high-income communities have the highest densities and so experience the strongest pressures to conform. We argue that there is convincing evidence in support of this pattern. 

First, \cite{bailey2020social} finds as income increases, people have more geographically dispersed social networks -- meaning a lower proportion of their friends reside nearby. This relationship is present at all income levels: the higher someone's income, the more dispersed their friends are.  Under the assumption that social comparisons are stronger when friends are geographically closer -- as their consumption decisions are likely more closely observed -- this suggests that higher income individuals may experience weaker social comparisons. This results in lower centrality for higher-income communities.

Second, \cite{chetty2022social_1,chetty2022social_2} show significant homophily based on socioeconomic status, with individuals tending to befriend others with similar income levels, particularly at the upper end of the income distribution.\footnote{Homophily is the tendency for people to be disproportionately connected to those similar to themselves, and is one of the most ubiquitous features of social networks \citep{jackson2008, pin2016networks}.} We then assume that people are more likely to share an identity with those of a more similar socioeconomic status. A key feature of Chetty et al.'s finding is that the level of homophily rises sharply at the very top of the income distribution -- around the 90th percentile and up. This suggests that those at the very top of the income distribution may experience stronger social comparisons. 

Importantly, while geographic dispersion seems to affect all income levels, the increasing homophily based on income appears primarily among those with higher incomes. So the geographic dispersion result creates a negative relationship between incomes and community density (which is driven entirely by the strengths of social comparisons) for most of the income distribution. Then the strongly increasing homophily at the top of the income distribution may create a reversal, and a positive relationship in this range. Together, these dynamics produce a non-monotonic relationship. \Cref{fig:1} provides a stylised illustration.

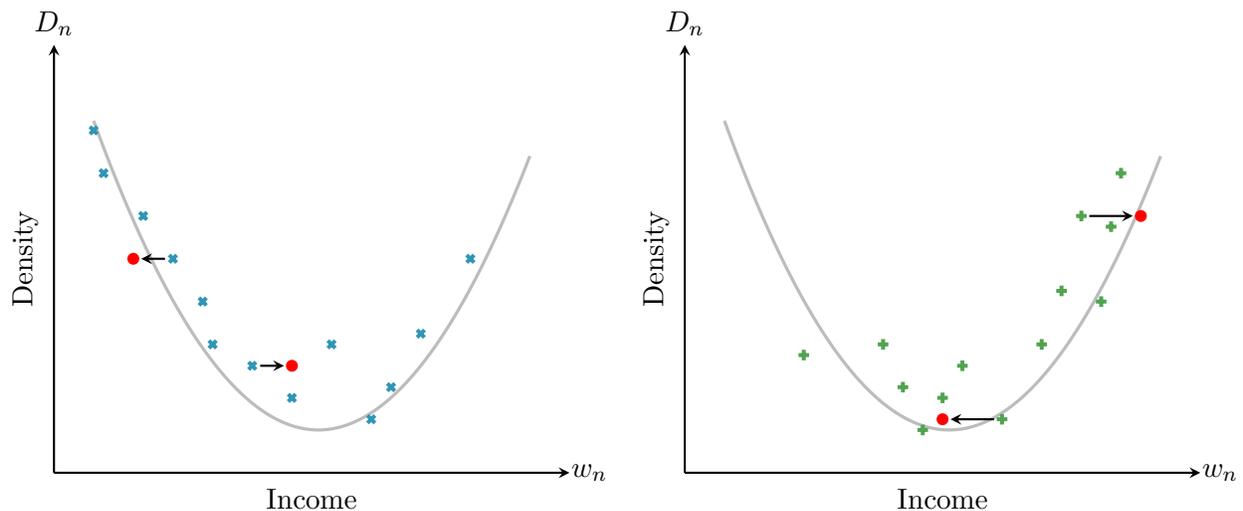
\begin{figure}[ht]
    \hspace{-1mm}
    \begin{subfigure}{0.45\textwidth}
        \centering
        \begin{tikzpicture}
\begin{axis}[
  axis lines=left,
  axis line style={thick},
  xlabel={Income},
  ylabel={Density},
  xtick={\empty},
  ytick={\empty},
  ymin=0,
  ymax=10, 
  xmin=0,
  xmax=13, 
  clip=false,
  domain=0:13, 
  samples=120 
]

\addplot+[no marks, very thick, gray, domain=1:12] {0.4*(0.75*x-5)^(2)+1};

\addplot+[mark=x, only marks, mark size=2pt, color=blue, mark options={line width=1.5pt}] coordinates {(1,8) (1.25,7) (2.25,6) (3,5) (2.25,4) (4,3) (5,2.5) (6,1.75) (7,3) (8,1.25) (8.5, 2) (9.25,3.25) (10.5, 5)};

\addplot+[mark=*, only marks, mark size=1.5pt, color=red, mark options={line width=1.5pt}] coordinates {(2,5) (6,2.5)}; 

\draw[thick, -stealth](2.8,5) -- (2.2,5);
\draw[thick, -stealth](5.2,2.5) -- (5.8,2.5);
\draw[very thick, -stealth, red](2.9,4.75) -- (5.1,2.75);

\node at (axis cs: 13.5,0) {$w_n$};
\node at (axis cs: 0,10.5) {$D_n$};

\end{axis}
\end{tikzpicture}
    \end{subfigure}
    \hspace{7mm}
    \begin{subfigure}{0.45\textwidth}
        \centering
        \begin{tikzpicture}
\begin{axis}[
  axis lines=left,
  axis line style={thick},
  xlabel={Income},
  ylabel={Density},
  xtick={\empty},
  ytick={\empty},
  ymin=0,
  ymax=10, 
  xmin=0,
  xmax=13, 
  clip=false,
  domain=0:13, 
  samples=120 
]

\addplot+[no marks, very thick, gray, domain=1:12] {0.4*(0.75*x-5)^(2)+1};

\addplot+[mark=+, only marks, mark size=2pt, color=green, mark options={line width=1.5pt}] coordinates {(3, 2.75) (5,3) (5.5,2) (6,1) (6.5,1.75) (7,2.5) (8,1.25) (9.75,2.75) (10,4) (10,6) (11,4.5) (10.75,5.5) (11,7)};

\addplot+[mark=*, only marks, mark size=1.5pt, color=red, mark options={line width=1.5pt}] coordinates {(6.5,1.25) (11.5,6)}; 

\draw[thick, -stealth](7.8,1.25) -- (6.7,1.25);
\draw[thick, -stealth](10.2,6) -- (11.3,6);
\draw[very thick, -stealth, red](7.9,1.5) -- (10.1,5.75);

\node at (axis cs: 13.5,0) {$w_n$};
\node at (axis cs: 0,10.5) {$D_n$};

\end{axis}
\end{tikzpicture}
    \end{subfigure}

    \caption{Gray line shows a possible relationship between income and density. Blue crosses (left panel) show possible realisations for Black communities. Green pluses (right panel) show possible realisations for White communities. Red arrows and red dots (both panels) show a change due to a regressive transfer between two communities.}
    \label{fig:1}
\end{figure}


\paragraph{Explaining the fact (Fact 5).} Assuming this non-monotonicity, we now consider the effect of a regressive transfer from one community to another (i.e. from a poor community to a richer one). This is a clean way of studying increases in inequality. Notice that the regressive transfer does not change the network structure -- only incomes. So it will not induce any changes to density. If this regressive transfer affects two higher-income communities, then it will shift income from a lower-density community to a higher-density community. 
It then follows from \Cref{prop:inequality 1}(ii) that consumption in communities with unchanged incomes will \emph{decrease}.

In contrast, if the regressive transfer affects two lower-income communities, then the effects will be reversed. The transfer will shift income from a higher-density to a lower-density community. This will in turn lead to an \emph{increase} in consumption in unaffected communities.

If such a non-monotonic relationship exists, we would expect that Blacks are disproportionately below the trough of the density curve, while Whites are disproportionately above it, reflecting the well-documented fact that Blacks tend to have lower average incomes than Whites \citep{charles2009conspicuous}. Thus, an increase in inequality among Whites would lead to lower consumption of visible (status) goods among Whites (whose incomes did not change), while the opposite effect would be observed for Blacks. This aligns perfectly with the findings of \citet{charles2009conspicuous}.

Our model can therefore explain this puzzling empirical finding without attributing it to any differences between racial groups other than well-established income disparities. 
we view this as a useful contribution because existing signalling models struggle to explain this empirical finding without appealing to racial differences in preferences.


\section{Discrimination}\label{sec:discrimination}
As a final exercise, we extend our model to include an exogenous component to group status. To avoid confusion, we will call this component \emph{prestige}. This extension explains our sixth stylised fact -- that even controlling for income, different groups spend different amounts on status goods \citep{charles2009conspicuous, kaus2013conspicuous}. This happens because higher prestige of an agent's group: (1) reduces her consumption of the status good, and (2) makes her consumption of the status good less responsive to her own income. So high prestige groups consume less status goods conditional on income. Given that discrimination is one determinant of prestige, our framework suggests that discrimination \emph{drives} higher spending on status goods. To our knowledge, this is a novel empirical prediction. 

\paragraph{Extending the model.} To build this into our framework, we simply adjust the definition of group status, $Y_\theta$. We add some exogenous prestige to the average consumption of each identity group. Formally,
\begin{align*}
    Y_\theta = \frac{\mathbf{x}_\theta + P_\theta}{\mathbf{x}_{-\theta} + P_{-\theta}}
\end{align*}
where $P_\theta > 0$ is the prestige of group $\theta \in \{A,B\}$. Notice that group status remains relative. It now comes from two sources -- consumption of the status good (an endogenous choice) and prestige (exogenous). Here, prestige is a catch-all term for everything other than consumption decisions that determines group status.\footnote{Prestige itself is implicitly relative. We could capture this formally by requiring that $P_A + P_B = \text{const.} > 0$. Doing so would reinforce our results -- it is clear that an increase in $P_A$ and a decrease in $P_B$ push in the same direction. But it burdens the algebra without adding any insight.}

The idea that there is an exogenous component to status is very much in line with the standard approach to models of social identity within economics \citep{akerlof2000economics, shayo2009model, shayo2020social}. Implicit within this idea of an exogenous component is some underlying social structure that ascribes more prestige to some groups than to others. Our view is that one part of this social structure is discrimination -- we will return to this shortly.\footnote{These views, that social structures matter, and that social status has both an exogenous component and a behaviour-driven component, is also present in sociology \citep{bobo1999prejudice, fiske2010envy, ridgeway2014status, ridgeway2022significance}.}

\paragraph{Changing prestige.} We consider how a change to an identity group's prestige affects consumption decisions by agents in both groups. The headline is straightforward: higher prestige for identity group $A$ reduces consumption by all agents in group $A$ and increases consumption by all agents in group $B$. Mechanically, this works in exactly the same way as a change in the average centrality of group $A$, except that no agent in group $A$ experiences an increase in her own centrality -- which would act as a countervailing force (see \Cref{rem:1}, \Cref{prop:slutsky} and associated discussion). This is why the result is unambiguous for both groups.

\begin{prop}\label{prop:comp_stat_prestige}
    An increase in group $\theta$'s prestige, $P_\theta$: \\
    (i) \ \ strictly decreases consumption for all agents in group $\theta$, \\
    (ii) \ strictly increases consumption for all agents in group $- \theta$, \\
    (iii) makes consumption less sensitive to own centrality for all agents in group $\theta$.
\end{prop}

The intuition is exactly the same as in \Cref{prop:2}. An increase in prestige raises the overall status of belonging to the group, $Y_\theta$. Due to the imperfect substitutability between status from own consumption and status from belonging to the group, this increase in $Y_\theta$ reduces the marginal benefit of own consumption. So agents do less of it. The fact that this change comes from prestige, rather than average consumption of the group does not matter. 
And of course, the impact on the other group is perfectly inverted. So agents in the other group consume more. 

\paragraph{Discrimination and the stylised fact (Fact 6).} It follows immediately from \Cref{prop:comp_stat_prestige} that differences in prestige can explain the headline finding from \cite{charles2009conspicuous}, that ``\emph{Blacks and Hispanics devote larger shares of their expenditure bundles to visible goods (clothing, jewellery, and cars) than do comparable Whites.}'' Hence \Cref{prop:comp_stat_prestige} explains our sixth stylised fact. \cite{charles2009conspicuous} also provide an explanation for this empirical regularity based on a signalling theory. But that explanation relies on particulars of the income distribution. In contrast, our explanation relies only on Blacks and Hispanics having lower prestige than Whites.

One determinant of prestige is the discrimination a group faces within a society \citep{oldmeadow2010social, fox2015denying}. Highly discriminated groups are, almost by definition, lower prestige ones (all else equal). So results about prestige can equally be phrased in terms of discrimination. The model predicts that more discrimination causes members of a group to consume more status goods, conditional on income.\footnote{In reality, we might expect discrimination to also affect incomes -- something not captured in our model (most plausibly, we think that greater discrimination would lead to lower income). So an observation that highly discriminated groups consume less status goods unconditionally is \emph{not} inconsistent with our model.}
This prediction about a causal link running from the discrimination a group faces to its status good consumption is, to the best of our knowledge, novel to economics. 

\paragraph{Implications for non-status goods.} While our model has only one good, it can be thought of as an approximation of a two-good model with a budget constraint and exogenous prices (see discussion in \Cref{sec:model_discussion}). In that two-good setting, higher discrimination would lead to a lower consumption of the non-status good -- as agents spend more of their budget on the status good in an effort to make up for their otherwise low status. 
The implications will depend critically on what the non-status good is. One plausible candidate for the non-status good is capital -- whether in the form of savings, physical or human capital. In that case, discrimination reduces an agent's investment in capital.\footnote{The notion that savings are a non-status good and are crowded out by increased consumption of status goods finds support in the empirical literature \citep{christen2005keeping, darku2014income, jaikumar2015conspicuous} -- although the impacts of discrimination have not yet been explored.} 
This interpretation is in line with empirical evidence that higher status good consumption by poor households in India reduces their calorie intake -- an important form of investment for households at the risk of malnourishment \citep{colson2022does}.\footnote{\cite{colson2022does} shows that higher inequality increases status good consumption, and reduces calorie intake, of poor households. This is easily rationalised by our model. If the change in inequality is across identity, then it is explained by \Cref{prop:2}(iv).  If the change in inequality is within identity, and the U-shaped relationship we posited in \Cref{sec:inequality} holds, then it is explained by \Cref{prop:inequality 1}.}

This suggests a new potential harm of discrimination: it causes agents to lose out on the long-term benefits of capital accumulation by pushing them towards greater status good consumption. While beyond the scope of our model, this could have important implications for welfare and for long-run inequality. Formal analysis would require us to model capital explicitly, and to take a stand on exactly how it benefits agents. We leave this to future work, and view this section as simply providing a novel link running from discrimination to consumption decisions.

\section{Conclusion}\label{sec:conclusion}
The pursuit of status is a powerful driver of consumption decisions. Many goods are consumed in part for the status they bring, not just their ‘intrinsic’ usefulness. This is well-understood in the context of luxury goods and housing, but extends even to food \citep{atkin2021we}. An impressive body of empirical work has documented a range of stylised facts regarding the consumption of these ‘status goods’. But existing theory has so far struggled to rationalise them as a coherent whole.

This paper fills that gap. We present a simple model where status comes from people’s own consumption decisions and from the social identity group they belong to, and where people also compare themselves to their neighbours. Importantly, there is some imperfect substitutability between these two sources of status. This can explain six of the most important stylized facts concerning conspicuous consumption. In doing so, we highlight the role of social identity as a source of status and hence as a driver of decision-making.

Beyond explaining existing findings, our model suggests important considerations for the design of economic and social policies. If a government wants to increase agents’ investments in certain non-status goods, it needs to be aware of the way in which social pressures can ‘divert’ additional income toward status goods. But the possibility of a substitution effect between different sources of status presents an intriguing new way of increasing these investments among certain disadvantaged groups. Specifically, increasing the status of the social identity itself can directly reduce status good consumption – in turn leading to more investments in other goods. This approach may prove preferable to directly providing financial assistance. And it suggests new benefits to combatting discrimination.
Additionally, policies aimed at altering social networks could help mitigate wasteful consumption. Different social groups typically have distinct goods that carry social status. For example, fashion items may serve as a status symbol for some groups, while sporting goods could be a status good for others. By integrating individuals from different social groups – reducing homophily in networks – social pressure to consume specific status goods may be reduced.

One feature of social identity that our model omits is the fact that individuals typically navigate multiple identities throughout their lives. For example, someone may consume fashion items to gain status among the ‘fashion conscious’, and consume high-quality sporting goods to gain status among the ‘sporty’. Although social psychologists have found that individuals often choose a single identity for a given situation, in reality people juggle many identities over time. Each identity can have its own status goods. These identities may evolve as agents interact within various social circles, which may or may not overlap. We omit this feature because we do not need it to explain our stylised facts. But we believe that adding this feature will generate novel and interesting predictions.

Overall, our approach offers a deeper understanding of the underlying mechanisms that contribute to persistent differences in spending on status good consumption across groups, as well as the broader impacts of these differences on well-being, wealth, and savings.



\newpage
\singlespacing
\bibliographystyle{abbrvnat}
\addcontentsline{toc}{section}{References}
\bibliography{bib}
\cleardoublepage
\onehalfspacing
\appendix
\numberwithin{equation}{section}
\numberwithin{thm}{section}
\numberwithin{prop}{section}
\numberwithin{defn}{section}
\numberwithin{rem}{section}
\numberwithin{lem}{section}
\numberwithin{cor}{section}
\section{Proofs}\label{appendix:proofs}
\paragraph{\Cref{prop:1}}
First, we find the best response function, holding $Y_{\theta_j}$ constant (the first derivative is linear and decreasing in $x_j$, and is positive when $x_j = 0$, so a unique solution exists). Note that we have substituted in $R_j = \alpha_j \sum_{k \in I(\theta_j)} g_{jk} x_k = \sum_{k \in I(\theta_j)} G_{jk} x_k = \sum_k \widehat{G}_{jk} x_k$.\footnote{Working with a weighted average set-up, $\alpha_j \sum_{k \in I(\theta_j)} g_{jk} x_k$, or a weighted sum set-up, $\sum_{k \in I(\theta_j)} G_{jk} x_k$, are mathematically identical. But the notation is simpler if we use a weighted sum.}
\begin{align*}
    x_j = \frac{w_j}{\beta w_j + 1} (S'_x(Y_{\theta_j}) + \beta \sum_k \widehat{G}_{jk} x_k) 
\quad \implies \quad 
    x_j - \frac{\beta w_j}{\beta w_j + 1} \sum_k \widehat{G}_{jk} x_k = \frac{w_j}{\beta w_j + 1} S'_x(Y_{\theta_j})
\end{align*}
Then recall that we defined the matrix $H$ such that $H_{jk} = \frac{\beta w_j}{\beta w_j + 1} \widehat{G}_{jk}$. We now substitute this in to the equation and switch to matrix notation. In an abuse of notation use the same notation for a vector of constants as for the equivalent scalar, and we let $\mathbf{W}$ denote the vector where the $j^{th}$ entry is $\frac{w_j}{\beta w_j + 1}$. And recall that $\odot$ is the Hadamard, or element-wise, product. Further, because we have assumed that $|\lambda_1| < 1$ (where $\lambda_1$ is the largest eigenvalue modulus of the matrix $H$. \Cref{ass:eigenvalue}), so $(I - H)$ is invertible \citep[Theorem 1]{ballester2006s}. Therefore:
\begin{align*}
    \mathbf{x} - H \mathbf{x} = \frac{1}{2} \mathbf{W} \odot S'_x(Y_{\theta})
\quad \implies \quad
    \mathbf{x} &= (I - H)^{-1} \cdot \mathbf{W} \odot S'_x(Y_{\theta}).    
\end{align*}
Reverting to scalar notation:
\begin{align*}
    x_j = \sum_k [(I - H)^{-1}]_{jk} \cdot \frac{w_k}{\beta w_k + 1} S'_x(Y_{\theta_k}).
\end{align*}
Now notice that, by construction, $\widehat{G}$ is a block matrix with off-diagonal blocks being all zero. This is because an agent $j$ can only have a link to agent $k$ if $j$ and $k$ share the same identity (i.e. if $\theta_j = \theta_k$). Therefore $H$ is also a block matrix. As a result, $H_{jk}$ will be zero whenever $j$ and $k$ have different identities (i.e. whenever $\theta_j \neq \theta_k$). An important consequence of this is that $\theta_k = \theta_j$ whenever $[(I - H)^{-1}]_{jk} > 0$. This means we can treat the term $Y_k$ as a constant and pull it out of the sum. So
\begin{align*}
    x_j = S'_x(Y_{\theta_j}) \cdot \sum_k [(I - H)^{-1}]_{jk} \cdot \frac{w_k}{\beta w_k + 1},
\end{align*}
and then using \Cref{defn:bonacich centrality}  
gives us:
\begin{align}\label{eq:fxpt1}
    x_j = S'_x(Y_{\theta_j}) \cdot C_j.
\end{align}
Recall that $Y_{\theta_j} = \frac{\mathbf{x}_A}{\mathbf{x}_B}$ if $\theta_j = A$ (and similarly if $\theta_j = B)$ and so is endogenous, where $\mathbf{x}_{\theta} = \frac{1}{|I(\theta)|} \sum_{k \in I(\theta)} x_k$. 
In an abuse of notation, we will use $Y_A$ to denote $Y_{\theta_j}$ for $j$ in identity $A$ (and similarly for $Y_B$). We will also do the rest of the proof for some $j$ in identity $A$. It is clear that everything works identically for identity $B$. 
Hence, by substituting our equation for consumption into this definition:
\begin{align*}
Y_A \quad = \quad 
    \frac{ \frac{1}{|I(A)|} \sum_{k \in I(A)} S'_x(Y_A) C_k }{ \frac{1}{|I(B)|} \sum_{k \in I(B)} S'_x(Y_B) C_k } 
\quad = \quad 
    \frac{S'_x(Y_A) \mathcal{C}_A}{S'_x(Y_B) \mathcal{C}_B}.
\end{align*}
The second equality follows directly from the definition of an average. We can pull the terms $S'_x(Y_A)$ and $S'_x(Y_B)$ through their respective sums because the status associated with an identity is the same for all agents with that identity by definition. Then using the fact that $Y_A = 1/Y_B$, we get: 
\begin{align*}
    Y_A = \frac{S'_x(Y_A) \mathcal{C}_A}{S'_x \left(\frac{1}{Y_A} \right) \mathcal{C}_B}
\end{align*}
As $S'_x(\cdot)$ is strictly decreasing (by assumption), the RHS of this equation is strictly decreasing in $Y_A$. The LHS is strictly increasing in $Y_A$. Call this solution $Y_A^*$. Substituting this into \Cref{eq:fxpt1} yields the result.
\hfill \qed \\

\noindent Before going further, we show that \Cref{ass:eigenvalue} guarantees bounds on all agents' equilibrium consumption, $x_j^*$, and the equilibrium ratio of average consumptions, $Y_\theta^*$. 
\begin{lem}\label{lem:bounds}
    Assumption 1 implies that $x_j^* \in(0, \psi \rho]$ and $Y^*_{\theta_j} \in \left(\frac{1}{\psi (1+\beta)} , \psi (1+\beta) \right)$.
\end{lem}

We relegate the proof of this lemma to \Cref{OA:additional_proofs}.
Then recall that we set $X \geq \psi \rho$ and $Z \geq \psi (1+\beta)$ (where $X$ was the upper bound on $x$ for which we have $S'_x > 0$, and $Z$ was the upper bound on $Y$ for which we have $S'_Y > 0$, and $\frac{1}{Z}$ was the `lower bound' on $Y$ at which we assumed that $S'_x$ was equal to $\rho<\infty$). So this result tells us that we have $S'_x > 0, S'_Y > 0, S''_{xY} < 0$ \emph{in equilibrium}. These properties were not required to prove \Cref{prop:1}, but will play a role in later comparative statics results.

\paragraph{\Cref{rem:1}} Follows immediately from \Cref{prop:1} (noticing that $Y_A^*$ is clearly increasing in $\mathcal{C}_A$, decreasing in $\mathcal{C}_B$, and invariant in $C_j$. \hfill \qed 

\paragraph{\Cref{rem:2}} Recall that $C_j = \sum_k \sum_{t=0}^{\infty} [H^t]_{jk} \cdot \frac{w_k}{\beta w_k + 1}$. First, notice that $\frac{w_k}{\beta w_k + 1}$ is increasing in $w_\ell$ when $k = \ell$, and is otherwise unaffected. This weakly increases $C_j$.
Second, a well-known property of the matrix inverse $(I - H)^{-1} = \sum_{t=0}^{\infty} [H^t]$ is that $\sum_{t=0}^{\infty} [H^t]_{jk}$ is increasing in $H_{\ell m}$, and strictly so if and only if there exists a walk (in $H$) from $j$ to $\ell$.\footnote{To see this, express $\sum_{t=0}^{\infty} [H^t] = H^0 + H^1 + H^2 + H^3 + ...$ It is immediate that all elements of $\sum_{t=0}^{\infty} [H^t]$ are weakly increasing in $H_{ji}$ for all $j,i$. Additionally, notice that $[H^2]_{j \ell} = \sum_k H_{jk} \cdot H_{ki}$. Therefore, $[H^2]_{j \ell}$ is positive and strictly increasing in $H_{k \ell}$ if $H_{jk} > 0$ (i.e. there is a walk of length one from $j$ to $k$). So $\sum_{t=0}^{\infty} [H^t]_{jk}$ is strictly increasing in all of $k$'s links when $j$ has a walk of length 1 to $k$. 
Similarly, $[H^{t+1}]_{j \ell} = \sum_k [H^t]_{jk} \cdot H_{k \ell}$, and $[H^{t+1}]_{j \ell}$ is positive and strictly increasing in $H_{k \ell}$ if $[H^t]_{jk} > 0$ (i.e. there is a walk of length $t$ from $j$ to $k$). So if there is a walk of some length (some $t > 0$) from $j$ to $k$, then $\sum_{t=0}^{\infty} [H^{t+1}]_{j \ell}$ is strictly increasing in $H_{\ell m}$ for all $m$. }
By definition, $H_{\ell m} = \frac{\beta w_\ell}{\beta w_\ell + 1} \cdot G_{\ell m} \cdot \mathbf{1} \{ \theta_\ell = \theta_m \}$. So $H_{\ell m}$ is increasing in $w_{\ell}$ and in $G_{\ell m}$.
Therefore $\sum_{t=0}^{\infty} [H^t]_{jk}$ is increasing in $w_\ell$ and $G_{\ell m}$, and strictly so if and only if there exists a walk (in $H$) from $\ell$ to $j$.

Note that $H_{\ell m}$ is strictly increasing in $G_{\ell m}$ whenever $\theta_\ell = \theta_m$. But for $H_{\ell m}$ to be strictly increasing in $w_\ell$, we need $\widehat{G}_{\ell m} > 0$. This need not be true for any given $m$. But if there exists a walk from $\ell$ to $j$, then it must be the case that $\widehat{G}_{\ell m} > 0$ for some $m$ (at a minimum, the link used to create the walk to $j$). So, if there exists a walk (in $H$) from $\ell$ to $j$, an increase in $w_\ell$ must strictly increase $H_{\ell m}$ for some $m$, and so must strictly increase $\sum_t [H^t]_{jk}$.
Finally, for there to be a walk in $H$ from $\ell$ to $j$, we must have $\theta_\ell = \theta_j$, because $H$ cannot contain any cross-identity links (by construction, it contains only within-identity links). \hfill \qed


\paragraph{\Cref{prop:slutsky}}
\textbf{(i)} If $\theta_j \neq \theta_k$, then $C_j$ and $\mathcal{C}_{\theta_j}$ do not vary with $w_k$ (\Cref{rem:2}). But $\mathcal{C}_{-\theta_j}$ increases, because $C_\ell$ is increasing for all $\ell$ such that $\theta_k = \theta_\ell = \theta'$, and strictly so for at least some agents (\Cref{rem:2}). The result then follows from \Cref{rem:1} -- as $x_j^*$ is strictly increasing in $\mathcal{C}_{-\theta_j}$. 

\textbf{(ii)} If $\theta_j = \theta_k$, then both $C_j$ and $\mathcal{C}_{\theta}$ are weakly increasing in $w_k$, and $\mathcal{C}_{-\theta_j}$ does not vary with $w_k$ (\Cref{rem:2}). Now take a partial derivative of \Cref{eq:closed form solution} with respect to $w_k$:
\begin{align}\label{eq:slutsky_proof}
    \frac{d x_j^*}{d w_k} &= 
    \underbrace{ C_j }_{>0} \cdot 
    \underbrace{ \frac{d S'_x(Y^*_{\theta_j})}{d Y^*_{\theta_j}} }_{<0} \cdot 
    \underbrace{ \frac{d Y^*_{\theta_j}}{d \mathcal{C}_\theta} }_{>0} \cdot 
    \underbrace{ \frac{d \mathcal{C}_\theta}{d w_k} }_{\geq 0} 
    + \underbrace{ S'_x(Y^*_{\theta_j}) }_{>0} \cdot 
    \underbrace{ \frac{d C_j}{d w_k} }_{\geq 0}
%
\end{align}

\paragraph{\Cref{prop:2}}
By \Cref{defn:communities_network}, each community $n$ is a separate component of the network. So, for agents $k \in n$, their centrality, $C_k$, depends only on $H_k$ (the subgraph of $H$ containing only agents $j' \in n$), not on the full network $H$. Consider an increase in $w_k$ for some $k \in n$. 

\textbf{Obs 1:} $\frac{d S'_x(Y^*_{\theta_j})}{d Y^*_{\theta_j}} < 0$ and $ \frac{d Y^*_{\theta_j}}{d \mathcal{C}_{\theta_j}} > 0$. The first is by assumption. The second follows from the definition of $Y^*_{\theta_j}$.
\textbf{Obs. 2:} if $j \in n$, then $\frac{d C_j}{d w_k} > 0$. This follows from \Cref{rem:2} (and the fact that by \Cref{defn:communities_network}, there is a walk from $j$ to $j'$ for all $j, j' \in n$). 
\textbf{Obs. 3:} if $j \notin n$, then $\frac{d C_j}{d w_k} = 0$. This is because $C_j$ depends only on agents in the same community (see above), and $j,k$ are in different communities. 
\textbf{Obs. 4:} if $\theta_j = \theta_k$, then $\frac{d \mathcal{C}_{\theta_j}}{d w_k} > 0$. Note that $\mathcal{C}_{\theta_j}$ is the average centrality of agents with identity $\theta_j$. And agent $k$ has identity $\theta_j$, and so is in a community of other agents with identity $\theta_j$. Then it follows from \Cref{rem:2} that the centralities of all agents in that community are strictly increasing in $w_k$. And the centralities in all other communities are unaffected (this is Obs. 3). And so the average centrality for the overall identity is increasing in $w_k$. We can decompose it as: $\frac{d \mathcal{C}_{\theta_j}}{d w_k} = \frac{\phi}{N} \frac{d \mathcal{C}^n_{\theta_j}}{d w_k} = \frac{\phi}{N} \frac{1}{M} \sum_{j \in n} \frac{d C_j}{d w_k}$, where $\mathcal{C}^n_{\theta_j}$ denotes the average centrality in community $n$, $M$ is the number of agents in community $n$, and $\phi \in (0,1)$ is the fraction of communities that have identity $\theta_j$.

\textbf{Part (i)} follows from Obs. 1, Obs. 3 and Obs. 4. \textbf{Part (ii)} follows from Obs. 2 and Obs. 3 \emph{and} the fact that with many different communities with identity $\theta$, the impact of the increase in $C_{j'}$ must outweigh the impact of the increase in $\mathcal{C}_{\theta}$ (so $\frac{d C_j}{d w_k} >> \frac{d \mathcal{C}_\theta}{d w_k}$).
\textbf{Part (iii)} follows from \Cref{prop:slutsky}(i).
\textbf{Part (iv)} as $j,k$ are in different communities, we have $\frac{d C_j}{d w_k} = 0$ (Obs. 2). The remaining expression in \Cref{eq:slutsky_proof} is:
\begin{align*}
    \frac{d x_j^*}{d w_k} &= C_j \cdot 
    \frac{d S'_x(Y^*_{\theta_j})}{d Y^*_{\theta_j}} \cdot 
    \frac{d Y^*_{\theta_j}}{d \mathcal{C}_{\theta_j}} \cdot 
    \frac{d \mathcal{C}_{\theta_j}}{d w_k} < 0
\end{align*}
Now, use the fact that $\frac{d \mathcal{C}_{\theta_j}}{d w_k} = \frac{\phi}{N} \frac{d \mathcal{C}^n_{\theta_j}}{d w_k}$ to get
\begin{align*}
    \frac{d x_j^*}{d w_k} &= C_j \cdot \frac{1}{N} \cdot \psi \quad \text{where} \quad
    \psi = 
    \frac{d S'_x(Y^*_{\theta_j})}{d Y^*_{\theta_j}} \cdot 
    \frac{d Y^*_{\theta_j}}{d \mathcal{C}_{\theta_j}} \cdot 
    \frac{d \mathcal{C}_{\theta_j}^n}{d w_k} < 0
\end{align*}
So we have
\begin{align}\label{eq:prop3-pt4}
     \frac{d^2 x_j^*}{d w_k d w_j} &= \frac{\phi}{N} \left[ \frac{d C_j}{d w_j} \cdot \psi + C_j \cdot \frac{d \psi}{d w_j} \right]
\end{align}
We know that $\frac{d C_j}{d w_j} >0$ (by \Cref{rem:2}). 

So the final step is to show that $\frac{d \psi}{d w_j} \propto \frac{1}{N}$. This will ensure that, for large enough $N$, the first term in the square brackets (immediately above) dominates.
To lighten the notation in this next step, we will let $P = \frac{d S'_x(Y^*_{\theta_j})}{d Y^*_{\theta_j}}, Q = \frac{d Y^*_{\theta_j}}{d \mathcal{C}_{\theta_j}}, R^n = \frac{d \mathcal{C}_{\theta_j}^n}{d w_k}$. Using this:
\begin{align*}
    \frac{d \psi}{d w_j} &= \frac{d P}{d w_j} Q R^n + \frac{d Q}{d w_j} P R^n + \frac{d \left( \frac{d \mathcal{C}_{\theta_j}^n}{d w_k} \right)}{d w_j} P Q.
\end{align*}
Now notice that $\mathcal{C}_\theta^n$ is the average centrality in community $n$. So it does not depend on agents outside of community $n$ (either their incomes or their network position). Importantly, agent $j$ is not in community $n$ by assumption. Therefore $\frac{d \left( \frac{d \mathcal{C}_{\theta_j}^n}{d w_k} \right)}{d w_j} = 0$. Next, apply the chain rule to obtain:
\begin{align*}
    \frac{d \psi}{d w_j} =R^n \left[ 
    \frac{d P}{d Y^*_{\theta_j}} \frac{d Y^*_{\theta_j}}{d \mathcal{C}_{\theta_j}} \frac{d \mathcal{C}_{\theta_j}}{d w_j} Q 
    + \frac{d Q}{d \mathcal{C}_{\theta_j}} \frac{d \mathcal{C}_{\theta_j}}{d w_j} P \right] 
    = R^n  \frac{d \mathcal{C}_{\theta_j}}{d w_j} \left[ 
    \frac{d P}{d Y^*_{\theta_j}} \frac{d Y^*_{\theta_j}}{d \mathcal{C}_{\theta_j}} Q 
    + \frac{d Q}{d \mathcal{C}_{\theta_j}} P \right].
\end{align*}
Agent $j$ is not in community $n$, so denote their community $\hat{n}$. Then we have 
\begin{align*}
    \frac{d \psi}{d w_j} &= \frac{\phi}{N} \frac{d \mathcal{C}_{\theta_j}^n}{d w_k} \frac{d \mathcal{C}_{\theta_j}^{\hat{n}}}{d w_j}  \left[ 
    \frac{d P}{d Y^*_{\theta_j}} \frac{d Y^*_{\theta_j}}{d \mathcal{C}_{\theta_j}} Q 
    + \frac{d Q}{d \mathcal{C}_{\theta_j}} P \right].
\end{align*}
Hence $\frac{d \psi}{d w_j}$ is proportional to $\frac{1}{N}$. So for large enough $N$, the first term in the square brackets of \Cref{eq:prop3-pt4} dominates, and $\frac{d^2 x_j^*}{d w_k d w_j} < 0$. 

\begin{lem}\label{rem:centrality simple}
Suppose the network is a communities network (a la \Cref{defn:communities_network}), all communities are the same size and all agents in a given community, $n$, have the same income, $w_{n}$.

\noindent Then $C_j = w_{n} \cdot \sum_k (I - \widehat{G})^{-1}_{jk}$ for all $j \in n$.
\end{lem}

\paragraph{Proof.} Start with the definition of Bonacich centrality (\Cref{defn:bonacich centrality}). Then apply (a) $H_{ij} = \left[\frac{\beta w}{\beta w + 1} \widehat{G}\right]_{ij}$ for all $i,j$, and (b) the Neumann series representation: $(I - H)^{-1} = \sum_{t=0}^{\infty} H^t$. 
Next, notice that under \Cref{defn:communities_network}, each community is a disconnected component in the network $\widehat{G}$. Therefore, we can treat each community separately. So we can replace $\widehat{G}$ with $\widehat{G}_n$, which is the graph for community $n$ only. Then recall that all agents in community $n$ have the same income, $w_n$. This yields:
\begin{align*}
    C_{j} = \sum_{k} \left[ \sum_{t=0}^{\infty} \left( \frac{\beta w_n}{\beta w_n + 1} \widehat{G}_n \right)^t \ \right]_{jk} \cdot \frac{w_n}{\beta w_n + 1}.
\end{align*}
Now the $\frac{w_n}{\beta w_n + 1}$ term at the end (far right) can be pulled out the front, and the same term that appears within the infinite sum can be separated from the $\widehat{G}_n$ term, and can also be pulled out of the summation over $k$.
\begin{align*}
    C_j = \frac{w_n}{\beta w_n + 1}  \cdot \sum_{t=0}^{\infty} \left( \frac{\beta w_n}{\beta w_n + 1} \right)^t \cdot \sum_{k} \left[ \sum_{t=0}^{\infty} \widehat{G}^t \ \right]_{jk}
\end{align*}
Then apply the standard closed form solution for the infinite sum of a geometric series (which always exists because $\frac{w_n}{\beta w_n + 1} < 1$ by construction), and notice that the result (partially) cancels with the $\frac{w_n}{\beta w_n + 1}$ term. This yields:\footnote{Where the final inequality is by the Neumann series representation of a matrix inverse.}
\begin{align*}
    C_j = w_n \cdot \sum_{k} \left[ \sum_{t=0}^{\infty} \widehat{G}^t \ \right]_{jk} \ = \  w_n \sum_k (I - \widehat{G})^{-1}_{jk}.
\end{align*}
\hfill \qed

\paragraph{\Cref{prop:inequality 1}}
First, it follows from \Cref{rem:centrality simple} that centrality rises in community $n'$ and falls in community $n$. Next, we can express average Bonacich centrality of an identity $\theta$ as the sum over agents in each community, then summed over communities with identity $\theta$ (rather than the sum over agents directly), finally divided by $\tfrac{1}{2} J$ (the number of agents with identity $\theta$). That is:
\begin{align}
    \mathcal{C}_{\theta} = \frac{2}{J} \sum_{n \in I(\theta)} \sum_{j \in n} w_n \sum_{k} (I - \widehat{G}_n)^{-1}_{jk},
\end{align}
where the sum over $n \in I(\theta)$ is an abuse of notation denoting a sum over all communities that have identity $\theta$.
Then, we use the fact that income is the same for all agents in a given community $n$, and substitute in for \Cref{defn:island density}. This yields:
\begin{align}\label{eq:centrality_density}
    \mathcal{C}_{\theta} = \frac{2}{J} \sum_{n \in I(\theta)} w_n D_n.
\end{align}
\textbf{Observation.} It is then clear that a transfer from community $n$ to $n'$: (1) decreases the centrality of agents in community $n$, (2) increases the centrality of agents in community $n'$, and (3) increases average Bonacich centrality of identity $\theta$ if and only if $D_{n'} > D_n$.

\textbf{Part (i)} For the communities $n$ and $n'$, the effect of own centrality dominates the effect of average centrality (because there are many communities). Centrality, and hence consumption, rises [resp. falls] for community $n'$ [resp. $n$].

\textbf{Part (ii)} For the other communities, own centrality does not change. Only average centrality matters. The result then follows from the observation above, and the fact that consumption is decreasing in average centrality (see \Cref{rem:1}). \hfill \qed

\paragraph{\Cref{prop:comp_stat_prestige}}
We first provide the closed-form solution for $x_j^*$ with the expanded definition of $Y_{\theta_j}$. Follow the proof to \Cref{prop:1} exactly up to the point where $x_j = S'_x(Y_{\theta_j}) C_j$. Then use the new definition for group status: $Y_{\theta} = \frac{\mathbf{x}_{\theta} + P_{\theta}}{\mathbf{x}_{-\theta} + P_{- \theta}}$. Substituting in for $\mathbf{x}_\theta$ (and pulling the $S'_x(Y_{\theta_j})$ through the summation in exactly the same way as in the proof to \Cref{prop:1}) yields:
\begin{align*}
    Y_\theta = \frac{S'_x(Y_\theta) \mathcal{C}_\theta + P_\theta}{S'_x(Y_{-\theta}) \mathcal{C}_{-\theta} + P_{-\theta}}.
\end{align*}
As in \Cref{prop:1}, LHS is strictly increasing in $Y_\theta$, and RHS is strictly decreasing in $Y_\theta$. So there is a unique solution $Y^*_\theta$. Then notice that RHS is increasing in $P_\theta$ and decreasing in $P_{-\theta}$. Therefore $Y^*_{\theta}$ is increasing in $P_\theta$ and decreasing in $P_{-\theta}$. 
It is then clear that $d x_j^* / d P_{\theta_j} < 0$ (\textbf{part (i)}) and $d x_j^* / d P_{- \theta_j} > 0$ (\textbf{part (ii)}) for all $j$. Finally, recall that $\frac{d x^*_j}{d C_j} = S'_x(Y^*_{\theta_j}) > 0$. Hence, $\frac{d^2 x^*_j}{d C_j d P_{\theta_j}} = \frac{d S'_x(Y^*_{\theta_j}) }{d Y^*_{\theta_j}} \cdot \frac{d Y^*_{\theta_j}}{d P_{\theta_j}} < 0$ (\textbf{part (iii)}). \hfill \qed

\newpage
\section*{For Online Publication}
\section{Additional Proofs}\label{OA:additional_proofs}

\paragraph{Proof of \Cref{lem:bounds}} \textbf{Step 1: upper and lower bounds on centralities.} First, notice that $H_{jk} \leq G_{jk}$ for all $j,k$. Therefore, the largest eigenvalue modulus of $H$ is weakly smaller than the largest eigenvalue modulus of $G$. This follows from standard results about matrix inverses. Then, coupled with Assumption 1, $(I-H)^{-1}$ is well-defined, and $\left[(I - H)^{-1}\right]_{jk} < \infty$ for all $j,k$. Additionally, $\frac{w_k}{\beta w_k + 1} < \infty$ for all $k$. So we have $C_j < \infty$ for all $j$. And with a finite number of agents, maximum value of their centralities, $\psi$, must exist and must be finite.

By standard properties of a matrix inverse, $\left[(I - H)^{-1}\right]_{jj} = 1$ and $\left[(I - H)^{-1}\right]_{jk} \geq 0$ for all $k\neq j$. And it is easy to verify that $\frac{w_k}{\beta w_k + 1}$ is strictly increasing in $w_k$. Therefore $\frac{w_k}{\beta w_k + 1} \geq \frac{1}{\beta + 1}$ because $w_k \in [1,\infty)$ by assumption. Combining these two observations yields $C_j \geq \frac{1}{\beta+1}$ for all $j$. So $\frac{1}{\beta+1}$ is a lower bound on centralities. 

\textbf{Step 2: bounds on ratio of average centralities.} Next, consider the ratio of average centralities. It is clear that the `most extreme' scenario is that all agents in one identity group have the maximum centrality (of $\psi$) and all agents in the other identity group have the minimum centrality of $1/(1+\beta)$. Together, this means that $\mathcal{C}_\theta / \mathcal{C}_{-\theta} \in (\frac{1}{\psi (1+\beta)} , \psi (1+\beta) )$. 

\textbf{Step 3: bounds on $Y^*_{\theta}$.} 
We have from \Cref{prop:1} that $Y_{\theta}^* $ solves $Y_{\theta} = \frac{S'_x(Y_{\theta}) \mathcal{C}_{\theta}}{S'_x(\frac{1}{Y_{\theta}}) \mathcal{C}_{-\theta}}$. Note that LHS is increasing in $Y_\theta$, and RHS is decreasing in $Y_\theta$. The second part is because, by assumption, $S'_x(\cdot)$ is strictly decreasing. 
WLOG, label the identities $\hat\theta, \tilde\theta$ such that $\mathcal{C}_{\hat\theta} \geq \mathcal{C}_{\tilde\theta}$.
First, consider $Y_{\hat\theta}$. Suppose $Y_{\hat\theta} = 1$. Then $1 < \frac{S'_x(1) }{S'_x(1)} \frac{\mathcal{C}_{\hat\theta}}{\mathcal{C}_{\tilde\theta}}$. Therefore $Y^*_{\hat\theta} \geq 1$. Suppose $Y_{\hat\theta} = \frac{\mathcal{C}_{\hat\theta}}{\mathcal{C}_{\tilde\theta}}$. Then $$\frac{\mathcal{C}_{\hat\theta}}{\mathcal{C}_{\tilde\theta}} > \frac{S'_x(\frac{\mathcal{C}_{\hat\theta}}{\mathcal{C}_{\tilde\theta}}) }{S'_x(\frac{\mathcal{C}_{\tilde\theta}}{\mathcal{C}_{\hat\theta}})} \frac{\mathcal{C}_{\hat\theta}}{\mathcal{C}_{\tilde\theta}}.$$
This is because $S'_x(\cdot)$ is strictly decreasing, and so $\mathcal{C}_{\hat\theta} \geq \mathcal{C}_{\tilde\theta} \implies S'_x(\frac{\mathcal{C}_{\hat\theta}}{\mathcal{C}_{\tilde\theta}}) \leq S'_x(\frac{\mathcal{C}_{\tilde\theta}}{\mathcal{C}_{\hat\theta}})$. Therefore $Y^*_{\hat\theta} \leq \frac{\mathcal{C}_{\hat\theta}}{\mathcal{C}_{\tilde\theta}}$.
%
%
Identical logic yields $Y^*_{\tilde\theta} \leq 1$ and $Y^*_{\tilde\theta} \geq \frac{\mathcal{C}_{\tilde\theta}}{\mathcal{C}_{\hat\theta}}$.
Putting these together gives us bounds: $Y^*_\theta \in (\frac{\mathcal{C}_{\tilde\theta}}{\mathcal{C}_{\hat\theta}}, \frac{\mathcal{C}_{\hat\theta}}{\mathcal{C}_{\tilde\theta}})$. Then substituting in the bounds we have on $\mathcal{C}_\theta / \mathcal{C}_{-\theta}$ yields:
$$Y^*_\theta \in \left(\frac{1}{\psi (1+\beta)} , \psi (1+\beta) \right).$$

\textbf{Step 4. Bounds on $x_j^*$.} From \Cref{prop:1}, we have that $x_j^* = S'_x(Y_{\theta_j}^*) C_j$. And we know that $S'_x(Y_{\theta_j}^*) \in (0, \rho]$. The lower bound is by assumption, and the upper bound is from Step 3. And we know that $C_j \in (\frac{1}{\beta+1}, \psi]$. The bounds are from Step 1. Therefore we have $x_j^* \in (0, \psi \rho]$. \hfill \qed

\paragraph{Making a slightly stronger assumption.} In the proof to \Cref{lem:bounds}, we consider the most extreme scenario possible -- where all agents with one identity had the lowest centrality possible, and all agents in the other identity group had the maximum centrality. While it allows us to state our results under the weakest possible assumptions, it is grossly unrealistic. In reality, the ratio of the group-average centralities will be (much) less severe. If we are willing to assume that the ratio is at most $\kappa \geq 1$, then we can skip steps 1 and 2 of the proof to \Cref{lem:bounds}, and directly assume that $\mathcal{C}_\theta / \mathcal{C}_{-\theta} \in (\frac{1}{\kappa} , \kappa)$. Redoing step 3 with this new assumption will then yield $Y_\theta^* \in (\frac{1}{\kappa} , \kappa)$. There is an indirect effect on step 4. As we must now have $Y_\theta^* \geq \frac{1}{\kappa}$, we must have $x_j^* \leq S'_x(\frac{1}{\kappa}) \psi$.

This is more than a curiosity. A stronger assumption on $\mathcal{C}_\theta / \mathcal{C}_{-\theta}$ allows for stronger substitutability between own $x$ and $Y$ while maintaining our assumptions $S'_x > 0, S'_Y > 0, S''_{xY} < 0$. Consider our example function for $S(x,Y)$ from \Cref{sec:model}: $S(x,Y) = \alpha x + Y - \gamma x Y$. There, $\gamma$ captures the `strength' of the substitutabilities. Under \Cref{ass:eigenvalue}, we required $\gamma < \min\{\frac{\eta}{\psi (1+\beta)}, \frac{1}{\psi \rho}\}$, recalling that we needed $X \geq \psi \rho$ and $Z \geq \psi (1+\beta)$, where $\psi$ is the maximum centrality of any agent. So the substitutabilities need to be fairly weak.
In contrast, if we are willing to assume that $\mathcal{C}_\theta / \mathcal{C}_{-\theta} \in (\frac{1}{\kappa} , \kappa)$, then we only need $\gamma < \min\{\frac{\eta}{\kappa}, \frac{1}{\psi S'_x(\frac{1}{\kappa})} \}$.

\begin{lem}\label{lem:extra1}
    (i) Our measure of Centrality (\Cref{defn:bonacich centrality}) exists and is well defined if and only if \Cref{ass:eigenvalue} holds.
\end{lem}
\paragraph{Proof.} \textbf{Part (i)} Centrality is well-defined if and only if $\sum_{t=0}^{\infty} H^t$ converges. 
By the Neumann Series representation, we have $\sum_{t=0}^{\infty} H^t = (I - H)^{-1}$. 
    \citet[Theorem 1]{ballester2006s} then shows that this matrix inverse is well-defined and non-negative if and only if $|\lambda_1| < 1$.
%
\hfill \qed

\paragraph{Extension of \Cref{prop:inequality 1} from \Cref{ft:inequality_total_effect}}
This proof picks up from the end of the proof to \Cref{prop:inequality 1}, and uses arguments from there.
As a reminder, $i,j,k,\ell$ denote agents, and $m,n$ denote communities. It is clear that the change can only affect agents with identity $\theta$, so we focus only on this identity group. Total consumption of identity $\theta$ agents, which we denote $X_\theta$, is: 
\begin{align}
    X_{\theta} = \sum_{i \in I(\theta)} x_i^* &= 
    \sum_{m \in I(\theta)} \sum_{j \in m'} S'_x(Y_\theta^*) \cdot C_j \\
    &= S'_x(Y_\theta^*) \sum_{m \in I(\theta)} \sum_{j \in m'} C_j \\
    &= S'_x(Y_\theta^*)  \sum_{m \in I(\theta)} w_m D_m \label{eq:A9}
\end{align}
W.L.O.G. suppose that the transfer of income from community $n$ to $n'$ raises income by $\epsilon > 0$ for all $j \in n'$, and reduces income by $\epsilon$ for all $j \in n$.\footnote{The fact that all communities are of equal size ensures that the size of the increase in income for each $j \in n'$ is the same as the decrease in income for each $j \in n$}

\paragraph{Step 1: impact of transfer on $\sum_{m \in I(\theta)} w_m D_m$} It is clear that this transfer increases $\sum_{m \in I(\theta)} w_m D_m$ if and only if $D_{n'} > D_n$. To see this, formally let $w_m$ denote pre-transfer incomes, and $w'_m$ denote post transfer incomes. Then we have $\sum_{m \in I(\theta)} w'_m D_m = (w_n - \epsilon) D_n + (w_{n'} + \epsilon) D_{n'} + \sum_{m \neq n, n'} w_m D_m$. 

\paragraph{Step 2: impact of $ \sum_{m \in I(\theta)} w_m D_m$ on $X_{\theta}$.} To streamline notation, we will let $\phi := \sum_{m \in I(\theta)} w_m D_m$ for this step.
First, recall (from \Cref{prop:1}) that $Y^*_\theta$ solves 
\begin{align*}
    Y_\theta = \frac{S'_x(Y_\theta) \mathcal{C}_\theta}{S'_x \left(\frac{1}{Y_\theta} \right) \mathcal{C}_{- \theta}} = \frac{S'_x(Y_\theta) \frac{2}{J} \phi}{S'_x \left(\frac{1}{Y_\theta} \right) \mathcal{C}_{- \theta}} 
\end{align*}
Then the second inequality uses our definition of $\mathcal{C}_\theta$ from \cref{eq:centrality_density} and the definition of $\phi$ (immediately above). Rearranging yields:
\begin{align*}
    \frac{J}{2} Y_\theta S'_x \left(\frac{1}{Y_\theta} \right) \mathcal{C}_{-\theta} = S'_x(Y_\theta) \phi
\end{align*}
The LHS of this equation is strictly increasing in $Y_\theta$. And the RHS is strictly decreasing in $Y_\theta$, and strictly increasing in $\phi$. Therefore the value of $S'_x(Y_\theta) \phi$ at the solution to the equation must be strictly increasing in $\phi$. To see this, notice that increasing $\phi$ shifts the curve of the LHS up (when plotting the value of the expression against $Y_\theta$). This upwards shift must mean the LHS curve now intersects the RHS curve at a higher value. Therefore $S'_x(Y^*_\theta) \phi$ is strictly increasing in $\phi$.

Finally, substitute \cref{eq:centrality_density} and the definition of $\phi$ into \Cref{eq:A9}:
\begin{align}
    X_{\theta} &= S'_x(Y^*_{\theta}) \cdot \phi
\end{align}
So $X_\theta$ is strictly increasing in $\phi$. The result follows from this fact, coupled with the observation from step 1 that the transfer increases $\phi$ if and only if $D_{n'} > D_n$. \hfill \qed

\paragraph{Convexity of the centrality measure.} First, recall that by the Neumann series representation, we have $(I-H)^{-1} = \sum_{t=0}^{\infty} H^t$. And define $\mathcal{Z}_{ij}(H) = \sum_{t=}^{\infty} [H^t]_{ij}$. This writes the entries of $(I-H)^{-1}$ as a function of $H$. Second, take the derivative in the direction of some arbitrary matrix $A$. This yields: $D \mathcal{Z} (H)[A]= (I-H)^{-1} A (I-H)^{-1}$. 

Then take the second derivative. This yields: $D^2 \mathcal{Z}(H)[A,A]= (I-H)^{-1} A (I-H)^{-1} A (I-H)^{-1}$. Hence $D^2 \mathcal{Z}_{ij}(H)[A,A]=[(I-H)^{-1} A (I-H)^{-1} A (I-H)^{-1}]_{ij}$. Since $(I-H)^{-1}$ is a non-negative matrix and $A$ is non-negative, 
we must have $D^2 \mathcal{Z}_{ij}(H)[A,A] \geq 0$ for all $i,j$ and all non-zero $A$. 

So $\mathcal{Z}_{ij}(H) = \sum_{t=}^{\infty} [H^t]_{ij} = [(I - H)^{-1}]_{ij}$ is weakly convex in each entry of $H$. Finally, recall that centrality is defined as a weighted sum of entries of $(I-H)^{-1}$. Formally: $C_i = \sum_{j} \left[(I - H)^{-1}\right]_{ij} \cdot \frac{w_j}{\beta w_j + 1}$. Because all entries of $(I - H)^{-1}$ are weakly convex in entries of $H$, centrality must also be weakly convex in entries of $H$.

\newpage
\section{Keeping up with the Joneses: an alternative to Dissonance}\label{OA:KUWJ}
The model in \Cref{sec:model} assumes that agents feel some pressure to conform to their respective reference points, which is a benchmark for what constitutes `normal'. They experience dissonance costs from being either above or below their reference point. This approach is motivated by the Social Identity Theory literature, which views deviation from a reference point as in some sense costly or unpleasant, irrespective of whether this is a deviation above or below.

However, a large part of the economics literature on social comparisons makes an alternative assumption: that agents experience costs from being below their reference point, but receive benefits from being above it. This is often interpreted as pressure to `Keep up with the Joneses' (or more accurately, to get ahead of them). 
We show that our main results are robust to this alternative assumption.
In order to keep the model clean, we simplify the cost function by assuming that it is linear and that income is the same for everyone ($w_j = w$ for all $j$). We also impose a particular functional form on the status benefits. This is just for convenience.
So an agent's preferences are now:
\begin{align*}
    u_j &= \underbrace{ \alpha x_j + Y_{\theta_j} - \gamma x_j Y_{\theta_j} }_{Status} + SOC -\frac{1}{w} x_j. \\
    \text{ where }
    SOC &= 
    \begin{cases}
        \quad \beta \sqrt{x_j - R_j} \text{ if } x_j \geq R_j \\
        - \beta' \sqrt{R_j - x_j} \text{ if } x_j < R_j
    \end{cases}
    \text{ with } \beta > 0, \beta' > 0
\end{align*}

We must have $(\alpha - \frac{1}{w}) > 0$ for the problem to be interesting -- otherwise actions clearly explode to infinity. Analogously to the main text, we will also assume that $|\lambda_1(\widehat{G})| < 1$, where $\lambda_1(\widehat{G})$ is the largest eigenvalue modulus of the matrix $\widehat{G}$. And define Bonacich centrality: $C_j^{bon} = \sum_k (I- \widehat{G})^{-1}_{jk}$. Note that here we are using the `standard' definition on Bonacich centrality. This is due to the simplifying assumption that incomes are equal for all agents.
With this setup and machinery, we can again characterise equilibrium behaviour. 

\begin{prop}
There exists a unique equilibrium. For all $j$ equilibrium consumption is:
\begin{align}
    x_j^* = \frac{\beta^2}{4} \frac{1}{\left(\frac{1}{w} - \alpha + \gamma Y_{\theta_j}^* \right)^{2}} \cdot C_j^{bon}
\end{align}
where $Y_{\theta_j}^*$ is a strictly increasing function of $\frac{\mathcal{C}_{\theta}^{bon}}{\mathcal{C}_{-\theta}^{bon}}$, and $\mathcal{C}_{\theta}^{bon}$ is average Bonacich centrality in identity group $\theta$.
\end{prop}

The four key features of equilibrium behaviour are the same as in the main model. First, consumption is linearly increasing in own centrality. Second, it is decreasing in the average centrality of those who share the same identity as $j$. Third, own centrality acts multiplicatively on the term containing average centrality. Fourth, the definition of centrality depends only on the within-identity links (i.e. it depends on $\widehat{G}$, rather than on $G$). Note that the definition of centrality here is simpler because we have assumed all agents have the same income. 

The presence of these four key features means that the other results from \Cref{sec:results} will hold here. The exception is the first part of \Cref{rem:2}, because incomes no longer feature in the definition of centrality. 

This variant on the model can therefore also explain the stylised facts. This is suggestive that it is the key novelty of the model -- the multiple sources of status and the substitutability between them -- that is driving the results, and not the peculiarities of the dissonance function. 

\subsection{Proof}
The proof here follows \Cref{prop:1} fairly closely. Consider the marginal utility of consumption:
\begin{align*}
    \frac{d u_j}{d x_j} &= \alpha - \gamma Y_{\theta_j} + \frac{1}{2} \beta (x_j - R_j)^{-0.5} - \frac{1}{w}
\end{align*}

Notice that $\frac{1}{2} \beta (x_j - R_j)^{-0.5} \geq 0$ whenever it is well-defined. So we must have $\gamma Y_{\theta_j} - \alpha + 1/w > 0$ for a best response to exist. This implies that we must have $Y_{\theta_j} \geq (\alpha - 1/w) / \gamma$. Given this, we can now consider $j$'s best response.
\begin{align*}
    (x_j - R_j)^{-0.5} &= \frac{2}{\beta} \left( \gamma Y_{\theta_j} - \alpha + \frac{1}{w} \right)  \\
    \implies x_j - \sum_k \widehat{G}_{jk} x_k &= \frac{\beta^2}{4} \left( \gamma Y_{\theta_j} - \alpha + \frac{1}{w} \right)^{-2}
\end{align*}
%
We now switch to vector notation (in an abuse of notation we use the same notation for a vector of constants as for the equivalent scalar).
\begin{align*}
    \mathbf{x} - \widehat{G} \mathbf{x} &= \frac{\beta^2}{4} \left( \gamma Y_{\theta_j} - \alpha + \frac{1}{w} \right)^{-2}, \\
    \mathbf{x} &= (I - \widehat{G})^{-1} \left[ \frac{\beta^2}{4} \left(\gamma Y_{\theta_j} - \alpha + \frac{1}{w} \right)^{-2} \right].
\end{align*}
We have assumed $|\lambda_1| < 1$, so $(I - \widehat{G})$ is invertible (\citet[Theorem 1]{ballester2006s}). Reverting to scalar notation:
\begin{align*}
    x_j = \sum_k (I - \widehat{G})^{-1}_{jk} \frac{\beta^2}{4} \left( \gamma Y_{\theta_j} - \alpha + \frac{1}{w} \right)^{-2}.
\end{align*}
\emph{A comment:} by construction, $\widehat{G}$ is a block matrix with off-diagonal blocks being all zero. This is because an agent $j$ can only have a link to agent $k$ if $j$ and $k$ share the same identity (i.e. if $\theta_j = \theta_k$). An important consequence of this is that $\theta_k = \theta_j$ whenever $[(I - \widehat{G})^{-1}]_{jk} > 0$. This means we can treat the term $Y_{\theta_k}^{-1}$ as a constant and pull it out of the sum. And we can use the (standard) definition of Bonacich centrality. So
\begin{align*}
    x_j^* = \frac{\beta^2}{4} \left( \gamma Y_{\theta_j} - \alpha + \frac{1}{w} \right)^{-2} C_j^{bon}
\end{align*}

Recall that $Y_{\theta_j} = \frac{\mathbf{x}_A}{\mathbf{x}_B}$ if $\theta_j = A$ (and similarly if $\theta_j = B)$ and so is endogenous, where $\mathbf{x}_{\theta} = \frac{1}{|I(\theta)|} \sum_{k \in I(\theta)} x_k$. 
In an abuse of notation, we will use $Y_A$ to denote $Y_{\theta_j}$ for $j$ in identity $A$ (and similarly for $Y_B$). We will also do the rest of the proof for some $j$ in identity $A$. It is clear that everything works identically for identity $B$. 
Hence, by substituting our equation for consumption into this definition:
\begin{align*}
Y_A \quad = \quad 
    \frac{ \frac{1}{|I(A)|} \sum_{k \in I(A)} \frac{\beta^2}{4} \left( \gamma Y_A - \alpha + \frac{1}{w} \right)^{-2} C_j^{bon} }{ \frac{1}{|I(B)|} \sum_{k \in I(B)} \frac{\beta^2}{4} \left( \gamma Y_B - \alpha + \frac{1}{w} \right)^{-2} C_k^{bon} } 
\quad = \quad 
    \frac{\left( \gamma Y_A - \alpha + \frac{1}{w} \right)^{-2} \mathcal{C}_A^{bon}}{\left( \gamma Y_B - \alpha + \frac{1}{w} \right)^{-2} \mathcal{C}_B^{bon}}.
\end{align*}
where $\mathcal{C}_{\theta}^{bon}$ denotes the average Bonacich centrality of agents in identity group $\theta$. It is also clear from the definition that $Y_A = 1/Y_B$. Substituting this, and rearranging yields: 
\begin{align}
    \left( \gamma Y_A - \alpha + \frac{1}{w} \right) Y_A^{0.5} = \left( \gamma \frac{1}{Y_A}  \alpha + \frac{1}{w} \right) \cdot \left( \frac{\mathcal{C}_A^{bon}}{\mathcal{C}_B^{bon}} \right)^{0.5},
\end{align}
which further re-arranges to
\begin{align}\label{eq:unique_YA}
        \gamma Y_A^{\frac{5}{2}} + \left( \frac{1}{w} - \alpha \right) Y_A^{\frac{3}{2}} - \left( \frac{1}{w} - \alpha \right) \left( \frac{\mathcal{C}_A^{bon}}{\mathcal{C}_B^{bon}} \right)^{0.5} Y - \gamma \left( \frac{\mathcal{C}_A^{bon}}{\mathcal{C}_B^{bon}} \right)^{0.5} = 0.
\end{align}
It is straightforward to see that the Left-Hand Side of \cref{eq:unique_YA} is strictly convex and is negative when $Y_A = 0$. So there is a unique non-negative solution to \Cref{eq:unique_YA}. Denote this unique solution $Y_A^*$. Finally, because the term $\mathcal{C}_A^{bon} / \mathcal{C}_B^{bon}$ only enters negatively, it is also clear that the value of $Y_A$ that solves \Cref{eq:unique_YA} is strictly increasing in $\mathcal{C}_A^{bon} / \mathcal{C}_B^{bon}$. \hfill \qed

\section{Absolute Status}\label{OA:abs status}
Agents gain status from belonging to a group with high average consumption. An important feature of our model is that group status is \emph{relative}. What matters is average consumption of an agent's own group relative to the other group. This notion that status is relative is perhaps the standard view. 

Nevertheless, here we show what happens if status depends instead on \emph{absolute} group consumption. This generates one change in the model: an alternative definition of $Y_\theta$. With absolute status, we now have $Y_\theta = \mathbf{x}_\theta$. This change has very little impact on equilibrium play. All that changes is the definition of $Y^*_{\theta_j}$. Formally:
\begin{prop}
There exists a unique equilibrium. For all $j$, equilibrium consumption is:
    \begin{align}
    x_j^* = S'_x(Y_{\theta_j}^*) \cdot C_j, \text{  where  } Y_{\theta_j}^* \text{ solves } Y_{\theta} =  S'_x(Y_\theta) \mathcal{C}_\theta
\end{align}
\end{prop}

\noindent \emph{Proof.} Follows directly from the proof to \Cref{prop:1} \hfill \qed  \\

What it does do, however, is shut down the strategic interactions between different groups. Behaviour of agents in group $\theta'$ now has no impact on agents in group $\theta$. This induces two changes to the results in \Cref{sec:results}: in \Cref{rem:1} $xj^*$ is now invariant in $\mathcal{C}_{-\theta}$, and in \Cref{prop:slutsky}(i) if $\theta_j \neq \theta_k$ then $x_j^*$ is invariant in $w_k$. 

This means that the model with absolute status is unable to explain our third stylised fact (Fact 3). \Cref{prop:2}(iii) no longer holds with absolute status. Additionally, each group makes decisions independent of all other groups. It is through the relative notion of status that our model is linking the strategic decisions of different groups.

\end{document}